\documentclass[11pt]{article}
\usepackage{fullpage}

\usepackage{belov_paper}

\usepackage{youngtab}

\newcommand{\Adv}{\mathop{\mathrm{ADV}^{\pm}}}

\newcommand{\C}{{\mathbb{C}}}
\newcommand{\I}{{\mathbb{I}}}
\newcommand{\J}{{\mathbb{J}}}

\newcommand{\spn}{\mathop{\mathrm{span}}}

\newcommand{\cp}{{\hspace{-0.5pt}\textsc{cp}\hspace{-0.5pt}}}
\newcommand{\se}{{\hspace{-0.5pt}\textsc{se}\hspace{-0.5pt}}}
\newcommand{\qp}{{\hspace{-0.5pt}\textsc{q}\hspace{-0.5pt}}}

\newcommand{\sh}{\mathop{\mathrm{sh}}}

\newcommand{\oW}{W}
\newcommand{\oX}{X}
\newcommand{\oY}{Y}
\newcommand{\oZ}{Z}

\newcommand{\OO}{\mathrm{O}}

\newcommand{\specht}[1]{S^{#1}}

\usepackage{tikz}
\usetikzlibrary{trees}
\usetikzlibrary{decorations.pathmorphing}
\usetikzlibrary{decorations.markings}

\def\tGamma{\tilde\Gamma}

\begin{document}

\title{Adversary Lower Bounds for the Collision \\ and the Set Equality Problems}
\author{
Aleksandrs Belovs\thanks{Faculty of Computing, University of Latvia.}
\and
Ansis Rosmanis\thanks{SPMS, Nanyang Technological University and Centre for Quantum Technologies.}
}
\date{}

\maketitle

\newcommand{\cpp}{{\sc Collision}\xspace}
\newcommand{\sep}{{\sc Set Equality}\xspace}

\newcommand{\w}[2]{{\stackrel{#1}{(#2)}}}

\begin{abstract}
We prove tight $\Omega(n^{1/3})$ lower bounds on the quantum query complexity of the \cpp and the \sep problems, provided that the size of the alphabet is large enough.
We do this using the negative-weight adversary method.
Thus, we reprove the result by Aaronson and Shi, as well as a more recent development by Zhandry.
\end{abstract}

\section{Introduction}

In theory, quantum query complexity is surprisingly well-understood.  For each function, the (negative-weight) adversary method~\cite{hoyer:advNegative} gives a tight characterisation of its quantum query complexity~\cite{reichardt:advTight}.  
The adversary bound is a semi-definite optimisation problem (SDP). 
It performs nicely under composition~\cite{reichardt:advTight}, and it was used to prove a strong direct product theorem~\cite{lee:strongDirect}.  There is no similar theory known for randomised query complexity.

In practice, however, it is quite hard to construct a feasible solution to the adversary bound SDP that would be close to the optimal.
In many cases, the positive-weight version of the bound is used.  It is the original version of the bound developed by Ambainis~\cite{ambainis:adv}.  
This version uses combinatorial reasoning, and it is easy to apply.
Unfortunately, the positive-weight adversary is subject to some severe constraints like the certificate complexity barrier~\cite{spalek:advEquivalent, zhang:advPower} and the property testing barrier~\cite{hoyer:advNegative}.
The latter one, for instance, states that, if each positive and each negative input differ in at least $\eps$ fraction of the input variables, no positive-weight adversary can prove a lower bound better than $\Omega(1/\eps)$.

The proper negative-weight adversary method usually requires spectral analysis of the adversary matrix.  It does not have many applications so far.
It was used to prove tight lower bounds for formulae evaluation~\cite{reichardt:formulae}.  
Also, a recent line of development~\cite{spalek:kSumLower, belovs:onThePower} resulted in a relatively general construction of adversary bounds superseding the certificate complexity barrier:
For any family $\cS$ of $\OO(1)$-sized subsets of $[n]$, a relatively simple optimisation problem gives a tight lower bound on the problem of detecting whether the input string $x\in[q]^n$ contains a subset of elements that belongs to $\cS$ and whose sum is divisible by $q$.  
The upper bound emerges from the model of non-adaptive learning graphs~\cite{belovs:learning}.
And the dual to this model gives a lower bound.
In particular, this gives nearly tight lower bounds for the $k$-{\sc{}Sum} and the {\sc Triangle-Sum} problems.

For some other functions, lower bounds on quantum query complexity were obtained using the polynomial method~\cite{beals:pol}.  
This method is known to be suboptimal~\cite{ambainis:polVsQCC}, but it is not subject to the limitations of the positive-weight adversary.
This method works well for functions with a lot of symmetries.
A notable example of its application is for the \cpp and the \sep problems.

In the \cpp problem, one has to decide whether the input function $x\colon [2n]\rightarrow[q]$ is one-to-one or two-to-one, provided that one of the two cases holds. 
The \sep problem is a special case of \cpp with an extra promise that the function $x$ is one-to-one on both subdomains $[1..n]$ and $[n+1..2n]$.
The \cpp problem was defined in~\cite{brassard:collision}, and a quantum $\OO(n^{1/3})$-query algorithm for the problem was given (which also works for \sep).
The \sep problem was defined by Shi~\cite{shi:collisionLowerOriginal} and conjectured to be as hard as \cpp.
By the property testing barrier, no positive-weight adversary can give more than a trivial lower bound for any of these problems.  
Consequently, all known lower bounds for these problems were obtained using the polynomial method.

First, Aaronson and Shi~\cite{shi:collisionLower} proved an $\Omega(n^{1/3})$ lower bound for the \cpp problem.  The proof was later simplified by Kutin~\cite{kutin:collisionLower}.
Next, Midrij\=anis showed an $\Omega((n/\log n)^{1/5})$ lower bound for \sep using a combination of the positive-weight adversary and the polynomial methods~\cite{midrijanis:setEquality}.
The lower bound for the \sep problem was used by Aaronson and Ambainis in their proof of the polynomial relation between the randomised and the quantum query complexity of partial permutation-invariant functions~\cite{aaronson:structure}.
Recently, Zhandry~\cite{zhandry:setEquality} proved a tight $\Omega(n^{1/3})$ lower bound for the \sep problem using rather complicated machinery from~\cite{zhandry:howToConstruct} based on the polynomial method.  This automatically strengthened the results in~\cite{aaronson:structure}.

In this paper, we extend the range of applications of the adversary bound, and use it to prove tight lower bounds for the \cpp and the \sep problems given that the size of the input alphabet, $q$, is at least $\Omega(n^2)$.
Thus, we resolve Shi's conjecture affirmatively, which was done independently from Zhandry's work.

There are several reasons why better understanding of the adversary method can be beneficial.  
First, this method is always tight.
Second, known constructions of the adversary lower bounds follow through duality to the upper bounds.  This simplifies the construction of the lower bound.  The construction in this paper is also based on duality, as explained in the beginning of \rfO(sec:const).
Interestingly, the proofs for the \cpp and the \sep problems in this paper are almost identical, showing that the adversary method can be easier adapted for a specific function, as soon as a lower bound for a similar function is obtained.  
This is in contrast to the polynomial method, as more than ten years separated Shi's and Zhandry's results.  
Also, we note that it is the first, to our knowledge, application of the adversary method that supersedes the property testing barrier.
We hope that some of our ideas will be useful in proving lower bounds via the adversary method for further problems.

The paper is organised as follows.  In \rfO(sec:prelim), we define the adversary method and basics of representation theory.  In \rfO(sec:analysis), we give a general treatment of the \cpp and the \sep problems, without actually defining the lower bound. 
In \rfO(sec:learning), we describe the dual learning graph for these problems, which inspires our adversary construction.
In \rfO(sec:const), we construct the lower bound.  Proofs of two technical results from this section are given in \rf{sec:proofs}.

As a final note about organisation of our paper, let use address a possible question of why we prove lower bounds for both the \cpp and the \sep problems.
Indeed, since the latter is a special case of the former, a lower bound for the former is superfluous.
Our justification is as follows.  First, \cpp is a better-known problem than \sep.  And then, the proof for \sep is slightly more involved and contains all the necessary ingredients for the \cpp lower bound.  Thus, a reader mainly interested in the \cpp problem can attain the proof by merely ignoring all parts on \sep.

\section{Preliminaries} \label{sec:prelim}

For integers $m$ and $\ell$, let $[m..\ell]$ denote the set $\{m,m+1,\ldots,\ell\}$, and we often use shorthand $[m]$ for $[1..m]$.
For sets $\cI$ and $\cJ$, let $\cI^\cJ$ denote the set of functions from $\cJ$ to $\cI$.  
We use shorthand $\cI^m$ for $\cI^{[m]}$ and $\cI^{m..\ell}$ for $\cI^{[m..\ell]}$.
We may also think of $\cI^m$ as strings of length $m$.

For finite sets $\cI$ and $\cJ$, an $\cI\times \cJ$ matrix is a matrix whose rows are labelled by the elements of $\cI$ and whose columns are labelled by the elements of $\cJ$.
For such a matrix $A$, and $i\in\cI$ and $j\in\cJ$, we denote by $A\elem[i,j]$ its $(i,j)$-th entry.  
We use similar notation for vectors.
We use $A^*$ for the conjugate operator.
Let $\circ$ denote the Hadamard (entry-wise) product of matrices.
We use $\bI$ and $\bJ$ to denote the identity and the all-ones matrices, respectively, to which, if necessary, we attach a subscript to indicate the space on which the correpsonding operactor acts.
If $\{b_1,\dots,b_m\}$ is a basis of $\bC^m$, we say that a vector $v = \alpha_1b_1+\cdots+\alpha_mb_m$ \emph{uses} a basis vector $b_i$ if the corresponding coefficient $\alpha_i$ is non-zero.

For a linear operator $A\colon V\to U$, its \emph{singular-value decomposition} is $A = \sum_i \sigma_i u_i v_i^*$, where $\sigma_i$ are positive real numbers, $\{v_i\}$ form an orthonormal system in $ V$ and $\{u_i\}$ in $ U$.  The names for $\sigma_i$, $u_i$, and $v_i$ are \emph{singular value}, \emph{left-singular vector}, and \emph{right-singular vector}, respectively.
The span of the left-singular vectors is called the \emph{image} and the span of the right-singular vectors the \emph{coimage} of $A$.
The \emph{spectral norm}, $\norm|A|$, is the largest singular value.
Singular vectors corresponding to the largest singular value are called \emph{principal}.
Finally, by principal singular vectors of $A$ \emph{on a subspace} $\cH\subseteq  U$ we understand principal singular vectors of the operator $A\Pi$, where $\Pi$ is the orthogonal projector on $\cH$.

\subsection{Adversary Method} \label{sec:adv}

We identify decision problems with Boolean-valued functions $f\colon \cD\to \{0,1\}$ with domain $\cD\subseteq[q]^{m}$. 
We call the inputs in $f^{-1}(1)$ and $f^{-1}(0)$ \emph{positive} and \emph{negative}, respectively.
We are interested in quantum query complexity of decision problems.  For the definitions and basic properties refer to~\cite{buhrman:querySurvey}.  
In the paper, we only require the knowledge of the (negative-weight) adversary bound that we are about to define.

\begin{defn}\label{defn:advMatrix}
 An {\em adversary matrix}\footnote{Compared to the general definition of an adversary matrix~\cite{hoyer:advNegative}, we  consider only a quarter of the matrix, as this quarter completely specifies the whole matrix (see~\cite{spalek:kSumLower} for details).
 }
 for a decision problem $f$ is a non-zero $f^{-1}(1)\times f^{-1}(0)$ matrix $\Gamma$.
 For any $i\in[m]$, the $f^{-1}(1)\times f^{-1}(0)$ matrix $\Delta_i$ is defined by
\[
 \Delta_i[\![x,y]\!] = \begin{cases}0,&x_i=y_i;\\1,&x_i\neq y_i.\end{cases}
\]
\end{defn}

\begin{thm}[Adversary bound~\cite{hoyer:advNegative,lee:stateConversion,spalek:kSumLower}]
\label{thm:adversary}
In notation of \refdefn{advMatrix}, the quantum query complexity of the decision problem $f$ is $\Theta\sA[\Adv(f)]$, where $\Adv(f)$ is the optimal value
of the semi-definite program
\begin{subequations}
\label{eq:adv}
\begin{alignat}{3}
 &{\mbox{\rm maximise }} &\quad& \norm|\Gamma| \\ 
 &{\mbox{\rm subject to }} && \norm|\Delta_i\circ\Gamma|\leq 1\quad \mbox{ for all }i\in[m],
\end{alignat}
\end{subequations}
where the maximisation is over all adversary matrices $\Gamma$ for $f$.
\end{thm}
Note that we can choose any adversary matrix $\Gamma$ and scale it so that the condition
$\max_i \norm|\Delta_i\circ\Gamma|\leq 1$ holds.  Consequently, we often use a relaxed condition
 $\norm|\Delta_i\circ\Gamma|= \mathrm{O}(1)$ instead of $\norm|\Delta_i\circ\Gamma|\leq 1$. 

Precise evaluation of $\norm|\Delta_i\circ\Gamma|$ may be hard, but we can upper bound
$\norm|\Delta_i\circ\Gamma|$ using the following trick first introduced in~\cite{belovs:adv-el-dist} and later used in~\cite{spalek:kSumLower,belovs:onThePower,spalek:adv-array}.
If $A$ is a matrix of the same dimensions as $\Delta_i$, we use the notation $\Delta_i\diamond A$ for a matrix $B$ satisfying
 \(
  \Delta_i\circ B=\Delta_i\circ A
 \).
Or, we write $A\stackrel{\Delta_i}{\longmapsto} B$.
Note that $B$ is not uniquely defined, 
and our task will be to choose one that fits our needs.
Now, from the fact that
\[
\gamma_2(\Delta_i) = \max_B\big\{\norm|\Delta_i\circ B|:\norm|B|\leq 1\big\}\leq 2
\]
(see~\cite{lee:stateConversion} for this and other facts about the $\gamma_2$ norm),
it follows that
\[
  \norm|\Delta_i\circ A| = \norm|\Delta_i\circ(\Delta_i\diamond A)|
  \leq 2\norm|\Delta_i\diamond A|.
\]
Note that we can always choose $\Delta_i\diamond A = A$
 and
\[
 \Delta_i\diamond(\alpha'A'+\alpha''A'')
  =\alpha'(\Delta_i\diamond A') + \alpha''(\Delta_i\diamond A'').
\]

In order to show that $\norm|\Delta_i\circ\Gamma| = \mathrm{O}(1)$, it suffices to show that
 $\norm|\Delta_i\diamond\Gamma| = \mathrm{O}(1)$ for some $\Delta_i\diamond\Gamma$.
 That is, it suffices to show that there is a way to 
 modify the entries $\Gamma\elem[x,y]$ of $\Gamma$ with $x_i=y_i$ so that the 
 spectral norm of the resulting matrix is bounded by a constant.

\subsection{Representation Theory}
\label{sec:representation}

In this section, we introduce basic notions from representation theory of finite groups with special emphasis on the symmetric group.  For more background, the reader may refer to~\cite{curtis:representationTheory, serre:representation} for general theory, and to~\cite{james:symmetricGroup, sagan:symmetricGroup} for the special case of the symmetric group.

Assume $G$ is a finite group.  
The \emph{group algebra} $\C G$ is the complex vector space with the elements of $G$ forming an orthonormal basis, where the multiplication law of $G$ is extended to $\C G$ by linearity.
A (left) \emph{$G$-module}, also called a \emph{representation} of $G$, is a complex vector space $V$ with the left multiplication operation by the elements of $\C G$ satisfying the usual associativity and distributivity conditions.  The module $V$ is equipped with an inner product invariant under the multiplication by the elements of $G$.
A \emph{$G$-morphism} (or just morphism, if $G$ is clear from the context) between two $G$-modules $V$ and $W$ is a linear operator $\theta\colon V\to W$ satisfying $\theta(uv) = u\theta(v)$ for all $u\in\C G$ and $v\in V$.

A $G$-module is called \emph{irreducible} (or just irrep for irreducible representation) if it does not contain a non-trivial $G$-submodule.
An essential basic result in representation theory is the following
\begin{lem}[Schur's Lemma]
Assume $\theta\colon V\to W$ is a morphism between two irreducible $G$-modules $V$ and $W$. 
Then, $\theta=0$ if $V$ and $W$ are non-isomorphic, otherwise, $\theta$ is uniquely determined up to a scalar multiplier.
\end{lem}
Copies of non-isomorphic irreps in a fixed $G$-module $V$ are orthogonal.  
For any $G$-module $V$, one can define its \emph{canonical} decomposition into the direct sum of \emph{isotypic} subspaces, each spanned by all copies of a fixed irrep in $V$. If an isotypic subspace contains at least one copy of the irrep, we say that $V$ {\em uses} this irrep.

If $G$ and $H$ are finite groups, then the irreducible $G\times H$-modules are of the form $V\otimes W$ where $V$ is an irreducible $G$-module and $W$ is an irreducible $H$-module.  And the corresponding group action is given by $(g,h)(v\otimes w) = gv\otimes hw$, with $g\in G$, $h\in H$, $v\in V$, and $w\in W$, which is extended by linearity. 

\paragraph{Symmetric group.}
Let $\bS_L$ denote the \emph{symmetric group} on a finite set $L$, that is, the group with the permutations of $L$ as elements, and composition as the group operation.  If $m$ is a positive integer, $\bS_m$ denotes the isomorphism class of the symmetric groups $\bS_L$ with $|L|=m$.
Representation theory of $\bS_m$ is closely related to \emph{Young diagrams}, defined as follows.

A \emph{partition} $\lambda$ of an integer $m$ is a non-increasing sequence $(\lambda_1,\dots,\lambda_k)$ of positive integers satisfying $\lambda_1+\dots+\lambda_k = m$.  We denote this by $\lambda\vdash m$, or write $m = |\lambda|$.
A partition $\lambda = (\lambda_1,\dots,\lambda_k)$ is often represented in the form of a Young diagram that consists, from top to bottom, of rows of $\lambda_1,\lambda_2,\dots,\lambda_k$ boxes aligned by the left side.
For a partition $\lambda=(\lambda_1,\dots,\lambda_k)$ of $m$ and an integer $\ell\ge \lambda_1$, by $(\ell, \lambda)$ we denote the partition $(\ell,\lambda_1,\dots,\lambda_k)$ of $m+\ell$.

For each partition $\lambda\vdash m$, we assign an irreducible $\bS_m$-module $\specht{\lambda}$, called the \emph{Specht module}.  All these modules are pairwise non-isomorphic, and give a complete list of all the irreps of $\bS_m$.  We describe these modules following the classical approach by Young (see~\cite[Chapter 3]{james:symmetricGroup} or~\cite[\S 28]{curtis:representationTheory} for more detail).

From now on, we assume that aforementioned $L = [m]$ for concreteness.  Assume $\lambda\vdash m$.
A \emph{Young tableau} of shape $\lambda$ is a Young diagram of $\lambda$ with each box containing an integer from $[m]$, each integer used exactly once.
We use $\sh(t)$ for the shape of $t$.
As an example, the following is a Young tableau of shape $(3,2)$:
\[
\young(451,23).
\]
For any Young tableau $t$, we define the group of its \emph{row permutations} $R_t$ and the group of its \emph{column permutations} $C_t$ to consist of the permutations in $\bS_m$ that permute the elements within each row or column of $t$, respectively.
In our example above, $R_t = \bS_{\{1,4,5\}}\times\bS_{\{2,3\}}$ and $C_t = \bS_{\{2,4\}}\times \bS_{\{3,5\}}$. 
In the representation theory of the symmetric group, the elements
\begin{equation}
\label{eqn:Et}
R^+_{t} = \sum_{\pi\in R_{t}}\pi,\qquad
C^-_{t} = \sum_{\rho\in C_{t}}(\sgn\rho)\rho,\qqand
E_t = C^-_{t}R^+_{t},
\end{equation}
of the group algebra $\bC\bS_m$ are widely used.  (Here $\sgn\rho$ stands for the sign of the permutation $\rho$.)
Suitably scaled multiples $\tR^+_t$, $\tC^-_t$, and $\tE_t$ of these operators are non-zero idempotents: 
$(\tR^+_t)^2 = \tR^+_t$, $(\tC^-_t)^2 = \tC^-_t$, and $\tE_t^2 = \tE_t$.
Note that, while $\tR^+_t$ and $\tC^-_t$ are orthogonal projectors, $\tE_t$, in general, is \emph{not} an orthogonal projector.

For any Young tableau $t$, the left ideal $\bC\bS_m E_t$ in the group algebra $\bC\bS_m$ is an irrep of $\bS_m$.  
The isomorphism class of $\bC\bS_m E_t$ only depends on the shape of $t$, and we denote by $\specht\lambda$ an arbitrary representative from this class, where $\lambda = \sh(t)$.
There exists a non-zero vector $v\in\specht\lambda$ such that $\tE_tv=v$ (in $\bC\bS_m E_t$ one can take $v = E_t$).
On the other hand, if $\mu\ne\lambda$ is another partition of $m$, then $\tE_t$ annihilates $\specht\mu$: $\tE_tv = 0$ for all $v\in\specht\mu$.

It  is not hard to see that $\pi R_t^+ = R_{\pi t}^+ \pi$, where $\pi t$ is the Young tableau $t$ with entry $i\in[m]$ replaced by $\pi(i)$.  
As a consequence, the module $\specht\lambda$ is spanned by the images of $R_t^+$ as $t$ runs through all the Young tableau of shape $\lambda$:
$\specht\lambda = \spn\sfig{ R_t^+(\specht\lambda) \mid \sh(t) = \lambda }$.

\section{Analysis of the Problem}
\label{sec:analysis}
Consider the following three types of input strings $x=(x_1,\dots,x_{2n})$ in $[q]^{2n}$:
\itemstart
\item[(a)] For each $a\in[q]$, there are either exactly two or none $i\in[2n]$ such that $x_i = a$.
\item[(b)] For each $a\in[q]$, either there is no $x_i$ equal to $a$, or there is unique $i\in[1..n]$ and unique $j\in[n+1..2n]$ satisfying $x_i=x_j=a$.
\item[(c)] For each $a\in[q]$, there is at most one $i\in[2n]$ such that $x_i=a$.
\itemend
In the \cpp problem, given $x\in[q]^{2n}$ satisfying (a) or (c), the task is to distinguish these two cases.  We say that (a) forms the set of positive inputs, and (c) is the set of negative inputs.  In the \sep problem, the task is to distinguish (b) and (c) in a similar manner.
Since any string in (b) also satisfies the requirements in (a), \sep is a simpler problem than \cpp.
In the following, we will use subscripts $\cp$ and $\se$ to denote relation to \cpp and \sep, respectively.  To avoid unnecessary repetitions, we use notation $\qp$ that may refer to both $\cp$ and $\se$.

Positive inputs in (a) and (b) naturally give rise to the corresponding \emph{matchings}.  A matching $\mu$ on $[2n]$ is a decomposition 
\[
[2n] = \{\mu_{1,1},\mu_{1,2}\} \cup \{\mu_{2,1},\mu_{2,2}\}\cup \cdots\cup \{\mu_{n,1}, \mu_{n,2}\}
\]
of the set $[2n]$ into $n$ mutually disjoint pairs of elements.  For concreteness, we will assume that $\mu_{i,j}$s are sorted: $\mu_{i,1} < \mu_{i,2}$ for all $i\in[n]$ and $\mu_{1,1}<\mu_{2,1}<\cdots<\mu_{n,1}$.  
In particular, $\mu_{1,1}=1$.
Clearly, this assumption is without loss of generality.  
Let $M_\cp$ denote the set of all matchings on $[2n]$, and let $M_\se$ denote the set of matchings $\mu$ on $[2n]$ such that $1\le \mu_{i,1}\le n$ and $n+1\le \mu_{i,2}\le 2n$ for all $i$.

\paragraph{Embedding.}
Our aim is to construct adversary matrices $\Gamma_\cp$ and $\Gamma_\se$ for the \cpp and the \sep problems, respectively. 
As described in \refdefn{advMatrix}, the rows of the adversary matrices are labelled by the positive inputs and the columns by the negative inputs.
We will use the trick initially used in~\cite{belovs:adv-el-dist}, and embed the adversary matrix $\Gamma_\qp$ into a larger $|M_\qp| q^n\times q^{2n}$ matrix $\tilde\Gamma_\qp$. 
The columns of $\tGamma_\qp$ are labelled by all possible inputs in $[q]^{2n}$.  The rows of $\tGamma_\qp$ are split into blocks corresponding to the matchings in $M_\qp$.  
Inside the block corresponding to a matching $\mu\in M_\qp$, there are all possible row labels $x\in[q]^{2n}$ such that $x\elem[\mu_{i,1}]=x\elem[\mu_{i,2}]$ for all $i$.
We will label the rows by specifying both the label and the block, i.e., as $(x, \mu)$.

A {\em legal} label of a row or a column of $\tGamma_\qp$ is one that also features in $\Gamma_\qp$.
Besides legal labels, $\tGamma_\qp$ also contains {\em illegal} labels.  A column is illegal if its label contains two equal elements.  A row corresponding to a matching $\mu$ is illegal if $x\elem[\mu_{i,1}] = x\elem[\mu_{j,1}]$ for some $i\ne j$.  
We obtain $\Gamma_\qp$ from $\tilde\Gamma_\qp$ by removing all the illegal rows and columns.
 
This embedding is used because it is easier to work with $\tilde\Gamma_\qp$ than with $\Gamma_\qp$.  We clearly have
\begin{equation}
\label{eqn:DeltaiNorms}
 \norm|\Delta_i\circ\Gamma_\qp| \leq \big\|\Delta_i\circ\tilde\Gamma_\qp\big\|
\end{equation}
because $\Delta_i\circ\Gamma_\qp$ is a submatrix of $\Delta_i\circ\tilde\Gamma_\qp$.
If we could show that $\norm|\Gamma_\qp|$ is not much smaller than $\|\tilde\Gamma_\qp\|$, that would allow us to use $\tilde\Gamma_\qp$ instead of $\Gamma_\qp$ in \refthm{adversary}. 
This is not true for every choice of $\tilde\Gamma_\qp$, however.
For instance, if $\tilde\Gamma_\qp$ contains non-zero entries only in illegal rows or columns, then
$\Gamma_\qp=0$.
But for our specific choice of $\tilde\Gamma_\qp$ this will be the case, as shown in \refsec{subIllegal}.
The condition $q = \Omega(n^2)$ is essential for the proof.

\paragraph{Space $\cH$ and its subdivision.}
Let $\cH = \C^{q}$, and let $\cM_\qp = \C^{M_\qp}$ be the complex vector space with the elements of $M_\qp$ forming an orthonormal basis.
Then, $\tilde\Gamma_\qp$ can be considered as an operator from $\cH^{\otimes 2n}$ to $\cM_\qp\otimes\cH^{\otimes n}$ if we identify a basis element $(\mu, z)\in M_\qp\times [q]^n$ with the row label $x$ in the $\mu$-block of $\tGamma_\qp$ defined by $x\elem[\mu_{i,a}] = z_i$.

Now we are going to define two bases of $\cH$.  The \emph{standard} basis of $\cH$ is the one used in \rfO(defn:advMatrix), in which the basis vectors correspond to the symbols of the input alphabet.  The {\em $e$-basis} is an orthonormal basis $e_0,e_1,\dots,e_{q-1}$ satisfying $e_0\elem[j]=1/\sqrt{q}$ for all $j\in[q]$.
The precise choice of the remaining basis vectors is irrelevant.  Further in the text, we almost exclusively work in the $e$-basis.
Let
\[
  \cH_0=\spn\{e_0\}\qquad\text{and}\qquad\cH_1 = e_0^\perp = \spn\{e_1,\ldots,e_{q-1}\}.
\]
Let us agree on a notational convention that $\Pi$ with arbitrary sub- and superscripts denotes the orthogonal projector onto the space denoted by $\cH$ with the same sub- and superscripts.
Thus, for instance, $\Pi_0=e_0e_0^*=\J_\cH/q$ and $\Pi_1=\I_\cH-\J_\cH/q$ are the projectors onto $e_0$ and its orthogonal complement, respectively. 
An important relation is
\begin{equation}
\label{eqn:Pis}
\Pi_1 \stackrel{\Delta}{\longmapsto} - \Pi_0,
\end{equation}
where $\Delta$ acts on the only variable.

Similarly, for the space $\cH^{\otimes m}$, the $e$-basis consists of all possible tensor products of the vectors in $\{e_i\}$ of length $m$.  Vector $e_0$ in the tensor product is called the {\em zero component}.  The {\em weight} of the basis vector is the number of non-zero components in the product.  The space $\cH^{\otimes m}$ can be decomposed as $\cH^{\otimes m} = \bigoplus_{k=0}^m \cH^{(m)}_k$, where
\[
  \cH^{(m)}_k=\bigoplus\nolimits_{c\in\{0,1\}^{m}\!, \,|c|=k}\cH_{c_1}\otimes\ldots\otimes\cH_{c_{m}}
\]
is the space spanned by all the basis elements of weight $k$.

\paragraph{Symmetry.}
Let $\bS_\cp=\bS_{[2n]}$ be the symmetric group on $2n$ elements, and let $\bS_\se$ be its subgroup $\bS_{[1..n]}\times \bS_{[n+1..2n]}$.
The \cpp and the \sep problems are invariant under the permutations of the input variables in $\bS_\cp$ and $\bS_\se$, respectively.  Hence, we may assume that $\Gamma_\qp$ is also invariant under these permutations~\cite{hoyer:advNegative}.  
We extend this symmetry to $\tGamma_\qp$ by requiring that, for each $\pi\in\bS_\qp$, and labels $(x,\mu)$ and $y$, we have
\begin{equation}
\label{eqn:symmetry}
\tGamma_\qp\elem[(x,\mu), y] = \tGamma_\qp\elem[(\pi x, \pi\mu), \pi y],
\end{equation}
where $(\pi x)_i = x_{\pi^{-1}(i)}$ and $\pi\mu = \sfigA{\{\pi(\mu_{1,1}), \pi(\mu_{1,2})\},\dots, \{\pi(\mu_{n,1}), \pi(\mu_{n,2})\}}$.  Note that $\cH^{(m)}_k$ is invariant under all the permutations in $\bS_{[m]}$.
Because of this symmetry, we may use the representation theory in the construction of $\tGamma_\qp$.

\section{Dual Learning Graph Perspective} \label{sec:learning}

Our lower bounds for the \cpp and the \sep problems are intrinsicly based on the dual learning graph for these problems developed in~\cite{belovs:onThePower}.  In this section, we explain the dual learning graph and how it relates to the adversary bound.

The concept of learning graphs is based on \emph{certificate structures}.  Informally, they describe all possible dispositions of certificates in positive inputs.  For \cpp and \sep, the corresponding certificate structure can be identified with the set of matchings $M_\qp$.
For each positive input $x$, there exists a matching $\mu\in M_\qp$ that pairs up equal elements.  That is, $x_a = x_b$ for all $\{a,b\}\in \mu$.
A subset $S\subseteq[2n]$ is a 1-certificate for $x$ if and only if it contains a pair from $\mu$ as a subset (that is, $2^S\cap\mu\ne\emptyset$).
Let us denote the latter relationship by $S\sim\mu$.

The (dual) learning graph complexity of the certificate structure is the optimal value of the following optimisation problem (the formulation is tailored to the case of \cpp and \sep): 
\begin{subequations}
\label{eqn:learningDual}
\begin{alignat}{3}
 &{\mbox{\rm maximise }} &\quad& \sqrt{\sum\nolimits_{\mu\in M_\qp} \alpha_\mu(\emptyset)^2} \label{eqn:alphaObjective} \\ 
 &{\mbox{\rm subject to }} && \sum_{\mu\in M_\qp} \sA[\alpha_{\mu}(S) - \alpha_\mu(S\cup\{j\})]^2\le 1 &\quad&\text{\rm for all $S\subseteq[2n]$ and $j\in[2n]\setminus S$;} \label{eqn:alphaOne} \\
 &&& \alpha_\mu(S) = 0 && \text{\rm whenever $S\sim\mu$;} \label{eqn:alphaZero}  \\
 &&& \alpha_\mu(S)\in\bR && \text{\rm for all $S\subseteq[2n]$ and $\mu\in M_\qp$.}
\end{alignat}
\end{subequations}
A \emph{dual learning graph} is any feasible solution to this optimisation problem. 
Let us note that one can also define the primal learning graph complexity, which is equal to the dual one, but we do not use it in this paper.

In some sense, learning graphs precisely capture quantum query complexity of certificate structures.  
First, for any decision problem with a given certificate structure, there exists a quantum algorithm for this problem with query complexity equal (up to a constant factor) to the learning graph complexity of the certificate structure.  
Second, for any certificate structure, there exists a decision problem possessing this certificate structure and whose quantum query complexity is equal (up to a constant factor) to the learning graph complexity of the certificate structure.

For our problems, we have the following

\begin{prp}
\label{prp:learning}
The dual learning graph complexity of the \cpp and the \sep problems is $\Omega(n^{1/3})$.
\end{prp}

\pfstart
Define the following potential solution to~\rf{eqn:learningDual}:
\begin{equation}
\label{eqn:alphamuS}
\alpha_\mu(S) =
\begin{cases}
\frac1{\sqrt{|M_\qp|}} \max\sfigA{ n^{1/3}-|S|, 0 }, & \text{if $S\not\sim \mu$;}\\
0,& \text{otherwise.}
\end{cases}
\end{equation}
It is easy to see that the objective value is $n^{1/3}$.  Let us prove its feasibility (up to a constant factor).
Fix $S\subseteq[2n]$ and $j\notin S$.  If $|S|\ge n^{1/3}$, the left-hand side of~\rf{eqn:alphaOne} is 0, so let us further assume $|S|< n^{1/3}$.
For each $\mu$, the difference $\alpha_\mu(S) - \alpha_\mu(S\cup\{j\})$ can take the following values:
\begin{itemize}
\item If $S\sim\mu$, then $\alpha_\mu(S) = \alpha_\mu(S\cup\{j\})=0$, and the difference is 0.
\item If $S\cup\{j\}\not\sim\mu$, then the difference is $1/\sqrt{|M_\qp|}$.
\item Finally, if $S\not\sim\mu$ and $S\cup\{j\}\sim\mu$, then the difference can be as large as $n^{1/3}/\sqrt{|M_\qp|}$.  However, this only happens if $j$ is matched by $\mu$ to an element already in $S$, and at most $|S|/n\le n^{-2/3}$ fraction of all matchings $\mu$ satisfy this condition.
\end{itemize}
Thus, the value of the left-hand side of~\rf{eqn:alphaOne} is at most
\[
|M_\qp|\cdot\frac1{|M_\qp|} + n^{-2/3} \cdot |M_\qp| \cdot \frac{n^{2/3}}{|M_\qp|} = \OO(1).
\]
By scaling down the solution by a constant factor, we obtain a feasible solution with objective value $\Omega(n^{1/3})$.
\pfend

Now we would like to convert this dual learning graph into an adversary lower bound.
In order to get some intuition, let us first consider the following $(M_\qp\times [q]^{2n})\times [q]^{2n}$ matrices.  Their rows are labelled by pairs $(\mu, x)$, where $\mu\in M_\qp$ and $x$ is an \emph{arbitrary} string in $[q]^{2n}$.  And their columns are labelled by strings in $[q]^{2n}$.
We treat such matrices as a $M_\qp\times 1$ block matrices, where each block is an $[q]^{2n}\times [q]^{2n}$ matrix.

\mycommand{hGamma}{\hat\Gamma}
Let $\hGamma$ be one such matrix whose $\mu$-th block is given by $\sum_S \alpha_\mu(S) \Pi_S$.
Here $\Pi_S$ is defined as $\bigotimes_{j\in[2n]} \Pi_{s_j}$, where $s_j = 1$ if $j\in S$ and $s_j = 0$ otherwise, and $\Pi_0$ and $\Pi_1$ are defined right before relation~\rf{eqn:Pis}.
Using~\rf{eqn:Pis}, we can define
\[
\Delta_j \diamond \Pi_S=
\begin{cases}
-\Pi_{S\setminus\{j\}},&\text{if $j\in S$;}\\
\Pi_S,&\text{otherwise.}
\end{cases}
\]
Since $\Pi_S$ are pairwise orthogonal projectors, we have
\[
\|\hGamma\|^2 = \max_S \sum_{\mu\in M_\qp} \alpha_\mu(S)^2
\qqand
\|\Delta_j\diamond\hGamma \|^2 = \max_{S\not\ni j} \sum_{\mu\in M_\qp} 
\sA[\alpha_{\mu}(S) - \alpha_\mu(S\cup\{j\})]^2.
\]
This is the intuition behind the expressions in~\rf{eqn:alphaObjective} and~\rf{eqn:alphaOne}.

Now we would like to switch from $\hGamma$ to $\tGamma$.
The simplest solution, adapted in~\cite{belovs:onThePower}, is to restrict $\hGamma$ to the rows used in $\tGamma$ and scale each block up, so that the norm of the thus transformed matrix $\Pi_\emptyset$ is still 1 in each block.
However, the norm of $\hGamma$ (and consequently, $\Delta_j\diamond\hGamma$) may grow after this transformation.  For example, consider a simple case of $n=1$ and two matrices
\[
\Pi_\emptyset = \Pi_0\otimes \Pi_0
\qqand
\Pi_{\{1\}} + \Pi_{\{2\}} = \Pi_1\otimes \Pi_0 + \Pi_0\otimes \Pi_1.
\]
The norms of both matrices are 1.  After the transformation, these matrices become
\begin{equation}  
\label{eqn:Psis}
\Psi_0=\Pi_0\otimes e_0^*=e_0^*\otimes \Pi_0
\qqand
\Psi_1=\Pi_1\otimes e_0^*+e_0^*\otimes \Pi_1.
\end{equation}
The norm of the first matrix is still 1, whereas the norm of the second is $\sqrt2$.  This growth in norm occurs essentially because two initially orthogonal left-singular vectors $e_i\otimes e_0$ and $e_0\otimes e_i$ collapse into one vector $e_i$, for any $e_i$ orthogonal to $e_0$.

Ideally, we would like the norm of $\Delta_j\diamond\hGamma$ not to grow more than by a constant factor.
The condition~\rf{eqn:alphaZero} is essential for this. 
In the lower bound constructions of~\cite{belovs:onThePower}, the decision problems possessing a given certificate structure were chosen in such a way that the norm did not grow much for any choice of $\alpha_\mu(S)$ satisfying~\rf{eqn:alphaZero}.
However, here we do not have a luxury of choosing the problems: we have to consider \cpp and \sep.  We will see in \rf{sec:subconst1} that the same simple restriction does not work for these problems.
In order to control the growth of the norm, we will have to switch to different operators, which we describe in \rf{sec:success}.

\section{Construction of the Adversary Matrix} \label{sec:const}
The matrix $\tilde\Gamma_\qp$ is constructed as a linear combination
\begin{equation}
\label{eqn:decomposition}
   \tilde\Gamma_\qp = \sum\nolimits_{k}\alpha_k\bar \oW_{\qp,k},
\end{equation}
where, for each $k$, $\bar W_{\qp,k}$ is an operator from $\cH^{(2n)}_k$ to $\cM_\qp\otimes \cH^{(n)}_k$.
The coefficients $\alpha_k$ are given by $\alpha_k = \max\{0, n^{1/3}-k\}$.
We again assume that $\bar \oW_{\qp,k}$ are invariant under the action of $\bS_\qp$ (in the sense of \rfO(eqn:symmetry)).

In terms of \rf{sec:learning}, one should think of $\bar W_{\qp,k}$ as corresponding to a combination of matrices $\Pi_S$ with $|S|=k$.
As mentioned in the previous section, the first intention is to use the techniques of~\cite{belovs:onThePower}
to construct the constituent matrices in~\rfO(eqn:decomposition). Unfortunately, this does not work, as we show in \rfO(sec:subconst1).
Luckily, it is possible to modify the construction of the matrices so that~\rfO(eqn:decomposition) gives an optimal adversary matrix.  We describe this in \rfO(sec:success).  Finally, in \rfO(sec:subIllegal), we show how to transform $\tGamma_\qp$ into a valid adversary matrix $\Gamma_\qp$.

\subsection{First Attempt} \label{sec:subconst1}
In this section, we define matrices $W_{\qp,k}$ that may seem as the most natural choice for the decomposition~\rfO(eqn:decomposition).  Unfortunately, they do not work well enough, so we will have to modify the construction in \rfO(sec:success).

As described in \rfO(sec:analysis),
the matrix $\tilde\Gamma_\qp$ can be decomposed into blocks corresponding to different matchings $\mu\in M_\qp$.  
We first define one block of the matrix.
Recall the operators $\Psi_0,\Psi_1\colon\cH^{\otimes2}\rightarrow\cH$ defined in~\rf{eqn:Psis}.
For every $k\in[0..n]$ and every $\mu\in M_\qp$, define the operator $W^\mu_k\colon \cH^{\otimes 2n}\to \cH^{\otimes n}$ by 
\begin{equation}\label{eqn:Fsigma}
 W^\mu_k = \sum\nolimits_{c\in\{0,1\}^n\!, \,|c|=k}\Psi_{c_1}\otimes\ldots\otimes\Psi_{c_n},
\end{equation}
where, for $i\in[n]$, $\Psi_{c_i}$ maps the $\mu_{i,1}$-th and the $\mu_{i,2}$-th multiplier in $\cH^{\otimes 2n}$ to the $i$-th multiplier in $\cH^{\otimes n}$.
Note that the image of $W^\mu_k$ is contained in $\cH^{(n)}_k$ and its coimage in $\cH^{(2n)}_k$, that is, $W^\mu_k = \Pi^{(n)}_k W^\mu_k \Pi^{(2n)}_k$.
The block of the matrix $W_{\qp, k}$ corresponding to $\mu\in M_\qp$ is defined by $\frac{1}{\sqrt{|M_\qp|}} W_k^\mu$.
We have $W_{\qp, k} = W_{\qp, k} \Pi^{(2n)}_k$.

Now, if we define $\tGamma_\qp$ as in~\rf{eqn:decomposition} with $\bar W_{\qp,k} = W_{\qp,k}$, we obtain the same matrix we would have obtained using the construction in \rf{sec:learning} with $\alpha_\mu(S)$ given by~\rf{eqn:alphamuS}.
One can see that $\oW_{\qp,k}$ thus constructed satisfy the symmetry~\rfO(eqn:symmetry).
Because of this, $\|\Delta_i\circ\tilde\Gamma_\qp\|$ is the same for all $i\in[2n]$.  Therefore, it suffices to estimate $\|\Delta_1\circ\tilde\Gamma_\qp\|$.
For that, we use the following simple decomposition
\begin{equation}
\label{eqn:PiDecomposition}
\Pi^{(2n)}_k = \Pi_0\otimes \Pi^{(2n-1)}_{k} + \Pi_1\otimes \Pi^{(2n-1)}_{k-1}.
\end{equation}
It is not hard to check that the $\mu$-th block of $W_{\qp,k}(\Pi_0\otimes \Pi^{(2n-1)}_{k})$ is given by $\frac{1}{\sqrt{|M_\qp|}}(X_k^\mu + Y_k^\mu)$ and the $\mu$-th block of $W_{\qp,k}(\Pi_1\otimes \Pi^{(2n-1)}_{k-1})$ is given by $\frac{1}{\sqrt{|M_\qp|}}Z_k^\mu$, where (with the same order of multipliers as in~\rf{eqn:Fsigma}):
\begin{equation}
\label{eqn:GHI}
\begin{split}
& X^{\mu}_{k} = \Psi_0\otimes \sum\nolimits_{c\in\{0,1\}^{2..n},\,|c|=k}\Psi_{c_2}\otimes\ldots\otimes \Psi_{c_{n}}, \\
& Y^{\mu}_{k} = (e_0^*\otimes \Pi_1)\otimes \sum\nolimits_{c\in\{0,1\}^{2..n},\,|c|=k-1}\Psi_{c_2}\otimes\ldots\otimes \Psi_{c_{n}}, \\
& Z^{\mu}_{k} = (\Pi_1\otimes e_0^*)\otimes \sum\nolimits_{c\in\{0,1\}^{2..n},\,|c|=k-1}\Psi_{c_2}\otimes\ldots\otimes \Psi_{c_{n}}.
\end{split}
\end{equation}
Here we used that $\mu(1,1)=1$, and also~\rf{eqn:Psis}.
Thus, if we define 
$\oX_{\qp,k}$, $\oY_{\qp,k}$, and $\oZ_{\qp,k}$ similarly to $W_{\qp, k}$, we get the following decomposition:
\begin{equation}
\label{eqn:Wdecomposition}
\oW_{\qp,k}=\oX_{\qp,k}+\oY_{\qp,k}+\oZ_{\qp,k}.
\end{equation}
Again, one can see that $\oX_{\qp,k}$, $\oY_{\qp,k}$, and $\oZ_{\qp,k}$
are symmetric under the action of $\bS'_\qp$, where $\bS'_\cp=\bS_{[2..2n]}$ and $\bS'_\se=\bS_{[2..n]}\times \bS_{[n+1..2n]}$. 
Using~\rf{eqn:Pis},
it is reasonable to define $\Delta_1\diamond \oX_{\qp,k}= \oX_{\qp,k}$,\; $\Delta_1\diamond \oY_{\qp,k}= \oY_{\qp,k}$,\; and $\Delta_1\diamond \oZ_{\qp,k}= -\oX_{\qp,k-1}$,
so that
\begin{equation}
\label{eqn:firstCancellation}
\tilde\Gamma_\qp \stackrel{\Delta_1}{\longmapsto} \sum\nolimits_k (\alpha_{k-1}-\alpha_k) \oX_{\qp,k-1} + \sum\nolimits_k \alpha_k \oY_{\qp,k} .
\end{equation}
This construction restates that of \rf{sec:learning}, where
the first and the second terms of this relation correspond to the second and the third bullets in the proof of \rf{prp:learning}, respectively.
In order to maintain a meaningful lower bound, we would need the norms of $W_{\qp,k}$ and $X_{\qp,k}$ to be $\OO(1)$, and the norm of $Y_{\qp, k}$ to be $\OO(\sqrt{k/n})$.  
However, in reality, it is not hard to show that
\begin{equation}
\label{eqn:exponentialNorm}
 \|W_{\qp, k}\| = \Theta(2^{k/2}),\qquad
\|X_{\qp, k}\| = \Theta(2^{k/2}),\quad\text{and}\quad\|Y_{\qp,k}\| = \Theta(2^{k/2}\sqrt{k/n}).
\end{equation}
This growth by the factor of $2^{k/2}$ can be interpreted as the $\sqrt2$ growth in~\rf{eqn:Psis} taken to the $k$-th power.  And this construction fails to give anything better than the trivial lower bound.

There is an explanation for this failure.
The construction above only used that the learning graph complexity of the \cpp problem is $\Omega(n^{1/3})$.  On the other hand, the learning graph complexity of the hidden shift problem is also $\Omega(n^{1/3})$~\cite[Proposition 12]{belovs:onThePower}.  Thus, if the current construction had worked, we would also have proven an $\Omega(n^{1/3})$ lower bound for the hidden shift problem contradicting the fact that the query complexity of this problem is logarithmic.
Thus, in order to obtain an optimal solution, we have to use again the structure of the problem.

\newcommand{\bbcH}{\bar{\bar\cH}}
\newcommand{\bbPi}{\bar{\bar\Pi}}

\subsection{Successful Construction}
\label{sec:success}
Our aim is to get rid of the $2^{k/2}$ factor in~\rfO(eqn:exponentialNorm) while preserving an analogue of~\rfO(eqn:firstCancellation).
Recall that $W^\mu_k = W^\mu_k \Pi^{(2n)}_k$.
As we will show in \rf{cor:Hm}, for $m$ a positive integer, $\cH^{(m)}_k$, as an $\bS_m$-module, uses irreps whose Young diagrams have at most $k$ boxes below the first row.
Let us define $\bar\cH^{(m)}_k$ as the subspace of $\cH^{(m)}_k$ spanned by the irreps of $\bS_m$ having \emph{exactly} $k$ boxes below the first row, i.e., of the form $(m-k,\lambda)$ with $\lambda\vdash k$.
We restrict each $W^\mu_k$ to this subspace, or, more precisely, we define
\begin{equation}\label{eqn:Wbar}
\bar \oW_{\qp,k}=\oW_{\qp,k}\bar\Pi_{\qp,k},
\end{equation}
where $\bar \Pi_{\qp,k}$ is the orthogonal projector on one of the following subspaces:
\begin{equation}\label{eq:barPiQP}
\bar\cH_{\cp,k}=\bar\cH_{k}^{(2n)} \qquad \text{or} \qquad \bar\cH_{\se,k}=\sum\nolimits_{\ell =0}^k\bar\cH_{k-\ell}^{(n)}\otimes \bar\cH_{\ell}^{(n)}.
\end{equation}
Here, for $\bar\cH_{\se, k}$, the first and the second multiplier reside in the first $n$ and the second $n$ copies of $\cH$ in $\cH^{\otimes 2n}$, respectively.

While applying $\Delta_1$ in \rf{sec:subconst1}, we used~\rf{eqn:PiDecomposition}, which effected to
\[\Pi^{(m)}_k = \Pi_0\otimes \Pi^{(m-1)}_{k} + \Pi_1\otimes \Pi^{(m-1)}_{k-1}.\]  
Now we would like to have a similar decomposition for $\bar\Pi^{(m)}_k$.  Unfortunately, this time there is a non-zero error term
\begin{equation}
\label{eqn:error}
\Phi_k^{(m)} = \bar\Pi_k^{(m)} - \Pi_0\otimes \bar\Pi_k^{(m-1)} - \Pi_1\otimes \bar\Pi_{k-1}^{(m-1)}.
\end{equation}
Note that $\Phi_k^{(m)}$ is a normal operator with image (and coimage) contained in $\cH^{(m)}_k$ and it is symmetric with respect to $\bS_{m-1}' = \bS_{[2..m]}$.
Luckily, in \rf{sec:newMainProof}, we will be able to bound it as follows:
\begin{lem}
\label{lem:newMain}
If $k < m/3$, then 
$\normA|\Phi_k^{(m)}| =  \mathrm{O}(1/\sqrt{m})$.  Moreover, the image of $\Phi_k^{(m)}$ only uses irreps of $\bS_{m-1}'$ with exactly $k-1$ boxes below the first row.
\end{lem}

With $\Phi^{(m)}_k$ as in~\rf{eqn:error}, let us define
\begin{equation}
\label{eqn:barHprime}
\begin{aligned}
  \bar\Pi_{\cp,k}' &=\bar\Pi_{k}^{(2n-1)},&
  \Phi_{\cp,k}&=\Phi_{k}^{(2n)},
  \\
  \bar\Pi_{\se,k}'&=\sum\nolimits_{\ell=0}^k(\bar\Pi_{k-\ell}^{(n-1)}\otimes \bar\Pi_{\ell}^{(n)}),&
  \Phi_{\se,k}&=\sum\nolimits_{\ell=0}^{k-1}(\Phi^{(n)}_{k-\ell}\otimes \bar\Pi_{\ell}^{(n)}).
\end{aligned}
\end{equation}

Note that 
$\Phi_{\qp,k}$ acts on $\cH^{\otimes2n}$ while $\bar\Pi'_{\qp,k}$ acts on $\cH^{\otimes(2n-1)}$.
From~\rf{eqn:error}, we have
\begin{equation}
\label{eqn:PiQkDecomposition}
 \bar\Pi_{\qp,k} = \Pi_0\otimes \bar\Pi'_{\qp,k} + \Pi_1\otimes \bar\Pi'_{\qp,k-1} + \Phi_{\qp,k}.
\end{equation}

With $X_{\qp,k}$, $Y_{\qp, k}$ and $Z_{\qp, k}$ as in \rfO(sec:subconst1), let
\[
 \bar \oX_{\qp,k}=\oX_{\qp,k}(\Pi_0\otimes\bar\Pi'_{\qp,k}),\qquad 
 \bar \oY_{\qp,k}=\oY_{\qp,k}(\Pi_0\otimes\bar\Pi'_{\qp,k}),\quad\text{and}\quad 
 \bar \oZ_{\qp,k}=\oZ_{\qp,k}(\Pi_1\otimes\bar\Pi'_{\qp,k-1}),
\]
so that from~\rfO(eqn:Wdecomposition) 
and~\rfO(eqn:PiQkDecomposition) we get
\[
 \bar \oW_{\qp,k}= \oW_{\qp,k} \bar\Pi_{\qp,k} = \bar\oX_{\qp,k}+\bar\oY_{\qp,k}+\bar\oZ_{\qp,k} + \oW_{\qp,k}\Phi_{\qp,k}. 
\]

We define the action of $\Delta_1$ on these operators by
\[
\bar \oX_{\qp,k} \stackrel{\Delta_1}{\longmapsto} \bar \oX_{\qp,k},\qquad
\bar \oY_{\qp,k} \stackrel{\Delta_1}{\longmapsto} \bar \oY_{\qp,k},\qquad
\oW_{\qp,k}\Phi_{\qp,k} \stackrel{\Delta_1}{\longmapsto} \oW_{\qp,k}\Phi_{\qp,k},
\quad\mbox{and}\quad
\bar \oZ_{\qp,k} \stackrel{\Delta_1}{\longmapsto} - \bar \oX_{\qp,k-1}.
\]
It is not hard to check that this definition satisfies the requirements of \rfO(sec:adv).  Thus, for $\tilde\Gamma_\qp$ as defined in \rfO(eqn:decomposition), we have
\[
\tilde\Gamma_\qp \stackrel{\Delta_1}{\longmapsto} \sum\nolimits_k (\alpha_{k-1}-\alpha_k) \bar \oX_{\qp,k-1} + \sum\nolimits_k \alpha_k \bar \oY_{\qp,k} +  \sum\nolimits_k \alpha_k \oW_{\qp,k}\Phi_{\qp,k}. 
\]
So far we have merely constructed an analogue of~\rfO(eqn:firstCancellation).  The main difference between this construction and the one in \rfO(sec:subconst1) is given by the following result, which we prove in \rf{sec:normEstimationsProof} (note the difference with~\rfO(eqn:exponentialNorm)):
\begin{lem}
\label{lem:normEstimations}
In the above notations, we have:
\[
(a)\quad \|\bar X_{\qp,k} \|\le 1,\qquad
(b)\quad \|\bar Y_{\qp,k} \| = \OO(\sqrt{k/n}),\qquad
(c)\quad \|W_{\qp,k}\Phi_{\qp,k} \| = \OO(1/\sqrt{n}).
\]
\end{lem}

With this result, it is not hard to  show that $\alpha_k=\max\{n^{1/3}-k,0\}$ is a good choice for the values of $\alpha_k$ in the decomposition~\rfO(eqn:decomposition). 
Indeed, for different $k$, all the operators $\bar \oX_{\qp,k}$ are mutually orthogonal, and the same is true for $\bar \oY_{\qp,k}$ and $\oW_{\qp,k}\Phi_{\qp,k}$.
Hence, the following conditions ensure that $\|\Delta_1\diamond \tilde\Gamma_\qp\| = \OO(1)$:
\[
|\alpha_{k-1}-\alpha_k|\leq 1,\qquad
|\alpha_k|\leq \sqrt{n/k},
\quad \text{and} \quad 
|\alpha_k|\le \sqrt{n}
\]
for all $k$.  
Our choice $\alpha_k=\max\{n^{1/3}-k,0\}$ satisfies these conditions, giving us
\[
\big\|\tilde\Gamma_{\qp}\big\|  \geq
\big\|\alpha_0\bar\oW_{\qp,0}\big\|  =
\alpha_0=n^{1/3}.
\]

\subsection{Removal of Illegal Rows and Columns} \label{sec:subIllegal}
So far we have only constructed the matrix $\tilde\Gamma_\qp$ in which the actual adversary
 matrix $\Gamma_\qp$ is embedded.
We obtain $\Gamma_\qp$ by deleting all the illegal rows and columns of $\tilde\Gamma_\qp$.
By~\rfO(eqn:DeltaiNorms), we already have that $\|\Gamma_\qp\circ\Delta_i\| = \OO(1)$ for all $i$.
It remains to show that $\|\Gamma_\qp\|$ is not much smaller than $\|\tilde\Gamma_\qp\|$, that is, not much smaller than $\alpha_0$.
Let us assume that $q\in\Omega(n^2)$, so that a constant ratio of rows and columns of $\tilde\Gamma_\qp$ remain in $\Gamma_\qp$ (that is, they are legal).

We have $\Gamma_\qp=\sum_k\alpha_k\check \oW_{\qp,k}$, where  $\check \oW_{\qp,k}$ is obtained form $\bar \oW_{\qp,k}$ by deleting all the illegal rows and  columns. 
In particular, since $\bar \oW_{\qp,0}=\oW_{\qp,0}$ is the matrix of all entries equal and $\|\bar \oW_{\qp,0}\|=1$, we have $\|\check \oW_{\qp,0}\|=\Omega(1)$ and its principal right-singular vector is the all-ones vector $\mathbf1$ of length $q!/(q-2n)!$.  All that is left to show is that $\check \oW_{\qp,k} \mathbf1=0$ whenever $k\neq 0$.

The coimage of $\bar W_{\qp, k}$ is contained in $\bar\cH_{\qp, k}$.  Let $\cL$ be the domain of $\check \oW_{\qp,k}$, which is spanned by the standard basis vectors corresponding to the legal negative inputs. 
Note that $\cL$ is a submodule of $\cH^{\otimes 2n}$, which is the domain of $\bar W_{\qp,k}$.  Hence, the coimage of $\check \oW_{\qp,k}$ is contained in the span of the irreps of $\bS_{\qp}$ used in $\bar\cH_{\qp, k}$.  
The vector $\mathbf1$ is contained in the trivial irrep of $\bS_{\qp}$ ($\specht{(2n)}$ for \cpp and $\specht{(n)}\otimes\specht{(n)}$ for \sep), and hence it is orthogonal to the coimage of $\check \oW_{\qp,k}$.  Thus, $\check \oW_{\qp,k} \mathbf1=0$ for $k>0$.
This gives us the main theorem of the paper.
\begin{thm}\label{thm:main}
For both $\qp\in\{\cp,\se\}$, let
\[
\tilde\Gamma_\qp=\sum\nolimits_{k=0}^{n^{1/3}}(n^{1/3}-k)\bar\oW_{\qp,k},
\]
where $\bar\oW_{\qp,k}$ is defined in \refeqn{Wbar}, and let $\Gamma_\qp$ be obtained from
$\tilde\Gamma_\qp$ by removing all the illegal rows and columns. 
Given that $q\in\Omega(n^2)$, $\Gamma_\cp$ and 
$\Gamma_\se$ are adversary matrices for \cpp and \sep, respectively,
giving an $\Omega(n^{1/3})$ lower bound on the quantum query complexity of both problems.
\end{thm}

\section{Proofs}
\label{sec:proofs}

\mycommand{skappa}{F}
In this section, we present the proofs of Lemmata~\ref{lem:newMain} and~\ref{lem:normEstimations}.
Throughout the whole section, we will work in the $e$-basis.
For the group algebra $\bC\bS_m$, we will use the standard basis consisting of the permutations $\bS_m$.

Recall the operator $E_t = C^-_tR^+_t$ from \rf{sec:representation} with the property that the submodule $\bC\bS_m E_t$ of the module $\bC\bS_m$ is isomorphic to $\specht{\lambda}$, where $\lambda$ is the shape of $t$.
Its scaled version $\tE_t$ is an idempotent.  Moreover, there exists a non-zero $v\in\specht\lambda$ with $\tE_tv = v$, whereas $\tE_tv=0$ for all $v\in\specht\mu$ with $\mu\ne\lambda$.
We start with constructing a similar operator for the whole space $\bar\cH_k^{(m)}$.
Let $\skappa \in \C\bS_m$ be defined by
\begin{equation}
\label{eqn:kappa}
\skappa = \frac{1}{2^{k}}\; \sA[\eps - (a_1,b_1)]\sA[\eps - (a_2,b_2)]\cdots \sA[\eps - (a_k, b_k)].
\end{equation}
Here, $a_1,b_1,\dots,a_k,b_k$ are some distinct fixed elements of $[m]$, $\eps$ is the identity element of $\bS_m$, and $(a_i,b_i)$ denotes the transposition of $a_i$ and $b_i$.
Note that $\skappa$ is an orthogonal projector.

\begin{lem}
\label{lem:kappa}
Let $\lambda\vdash m$. 
If $\lambda_1 = m-k$, then there exists a non-zero vector $v\in\specht\lambda$ such that $\skappa v = v$.  
If $\lambda_1 > m-k$, then $\skappa v = 0$ for all $v\in\specht\lambda$.
\end{lem}

\pfstart
We start with the first statement.
Let $\ell=\lambda_2$, and let $t$ be a Young tableau of shape $\lambda$ with $a_1,\dots,a_k$, in this order, being the first $k$ entries in the first row, and $b_1,\dots,b_k$ being the entries in the remaining rows, so that $b_1,\dots,b_{\ell}$ form the second row.
Since $\specht\lambda$ is isomorphic to $\bC\bS_m E_t$, we will work with the latter from now on.

Let us define
\begin{equation}
\label{eqn:vdef}
v = \sA[\eps - (a_{\ell+1},b_{\ell+1})]\cdots \sA[\eps - (a_k, b_k)] E_t \in\bC\bS_m E_t.
\end{equation}
Then, clearly, $\frac12 \sA[\eps - (a_i,b_i)] v = v$ for $i>\ell$.
On the other hand, $\frac12 \sA[\eps - (a_i,b_i)] v = v$ for all $i\le\ell$ as well since for them $(a_i,b_i)\in C_t$.  
Hence, we can conclude that $Fv = v$.

It remains to show that $v\ne 0$.
Open the brackets in~\rf{eqn:vdef}.  The result is $\sum_{\pi} (\sgn\pi)\pi E_t$, where 
$\pi$ runs through all possible products of the transpositions $(a_{\ell+1},b_{\ell+1}),\dots, (a_k,b_k)$.
Let $\rho = (a_{\ell+1},b_{\ell+1})\cdots(a_k,b_k)$.
Since the coefficient of the basis vector $\eps$ in $E_t$ is 1, the vector $\rho E_t$ uses the basis vector $\rho$.
Note that $\rho$ maps all $b_{\ell+1},\dots,b_k$ to the elements of $[m]$ lying outside of the first $\ell$ columns of $t$.
On the other hand, any permutation $\sigma$ used in $E_t$ maps $b_{\ell+1},\dots,b_k$ to the elements in the first $\ell$ columns of $t$.
Hence, $\pi\sigma$ maps some of $b_{\ell+1},\dots,b_k$ to an element in the first $\ell$ columns of $t$, unless $\pi$ is a product of all the $k-\ell$ transpositions, i.e., unless $\pi = \rho$.  
Thus, $\pi E_t$ does not use $\rho$ unless $\pi = \rho$. Hence, 
the coefficient of $\rho$ is non-zero in $v$, and, in particular, $v\ne 0$.

For the second statement, recall that any $\specht\lambda$ is spanned by the images of $R^+_t$ as $t$ runs through the Young tableau of shape $\lambda$.
Fix $t$ of shape $\lambda$.
Since $\lambda_1 > m-k$, at least one pair $\{a_1,b_1\},\dots,\{a_k,b_k\}$ is contained in the first row of $t$, and hence $\skappa$ is zero on the image of $R_t^+$.  Thus, $\skappa$ is zero on the whole $\specht\lambda$.
\pfend

\begin{cor}
\label{cor:Hm}
The $\bS_m$-module $\cH^{(m)}_k$ only uses irreps $\specht\lambda$ with $\lambda_1\ge m-k$.
\end{cor}

\pfstart
Assume towards contradiction that $\cH^{(m)}_k$ contains a copy of irrep $\specht\lambda$ with $\lambda_1<m-k$.  Then, by \rf{lem:kappa}, there exist distinct elements $a_1,b_1,\dots,a_{k+1},b_{k+1}\in [m]$ and a non-zero vector $v\in \cH^{(m)}_k$ such that $F''v=v$, where 
$
F'' = \frac{1}{2^{k+1}} \sA[\eps - (a_1,b_1)]\sA[\eps - (a_2,b_2)]\cdots\sA[\eps - (a_{k+1},b_{k+1})]
$.
On the other hand, for every $e$-basis vector $e_{c_1}\otimes\cdots\otimes e_{c_m}$ of weight $k$, there exists $i\in[k+1]$ such that $c_{a_i} = c_{b_i} = 0$.  Hence, $F''$ nullifies this basis vector, and, consequently, is zero on the whole $\cH^{(m)}_k$.
This contradiction finishes the proof.
\pfend

\begin{cor}
\label{cor:kappa}
If $v\in\cH^{(m)}_k$ is such that $\skappa v = v$, then $v \in \bar\cH_k^{(m)}$.
\end{cor}

\pfstart
By \rf{lem:kappa}, $v$ is orthogonal to all copies of $\specht\lambda$ in $\cH^{(m)}_k$ with $\lambda_1>m-k$.  And \rf{cor:Hm} then only leaves the possibility that $v$ is contained in $\bar\cH_k^{(m)}$.
\pfend

Before we proceed with the proofs, let us introduce the following piece of notation. 
For $a\in[m]$ and $z\in[0..q-1]$, $\w{a}{z}$ denotes the component $e_z$ residing in the $a$-th copy of $\cH$ in the tensor product $\cH^{\otimes m}$.
We might omit the upper index $a$ when $z=0$ and when it is clear in which copies of $\cH$ zero components $e_0$ must reside.
A juxtaposition of such components denotes the tensor product.  For instance,
given $m=6$,
\[
\w{4}{3}\w{2}{0}\w{5}{3}\w{1}{7}\w{}{0}\w{}{0}
= e_7\otimes e_0\otimes e_0\otimes e_3\otimes e_3\otimes e_0.
\]

\subsection{Proof of \rf{lem:newMain}}
\label{sec:newMainProof}
Let us for brevity denote $\Phi = \Phi^{(m)}_k$, and recall that 
\[
\Phi = \bar\Pi_k^{(m)} - \Pi_0\otimes \bar\Pi_k^{(m-1)} - \Pi_1\otimes \bar\Pi_{k-1}^{(m-1)}.
\]
Recall that $\Phi$ is symmetric with respect to $\bS_{m-1}' = \bS_{[2..m]}$, and it is zero on the orthogonal complement of $\cH^{(m)}_k$.
By the branching rule~\cite[Section 2.8]{sagan:symmetricGroup} applied to the term $\bar\Pi_k^{(m)}$ in the equation above, the image of $\Phi$ only uses irreps $\specht\lambda$ of $\bS_{m-1}'$ corresponding to partitions $\lambda\vdash m-1$ with $\lambda_1 = m-k$ and $\lambda_1 = m-k-1$.

Let us start with the latter case: $\lambda_1 = m-k-1$.
Fix $\lambda\vdash m-1$ with $\lambda_1 = m-k-1$.
By Schur's lemma, there exists a copy of $\specht\lambda$ in $\cH^{(m)}_k$ consisting of principal right-singular vectors of $\Phi$ on the isotypic subspace corresponding to $\lambda$.
Fix $k$ disjoint pairs $\{a_1,b_1\},\ldots,\{a_k,b_k\}\subset[2..m]$, and let them specify $F$ as in~\rf{eqn:kappa}.
By \rf{lem:kappa}, there is a vector $v$ in this irrep that satisfies $\skappa v = v$.

Consider the vector $v$ in the $e$-basis.
Since $F$ is an orthogonal projector, $v$ only uses those basis vectors that are not nullified by $F$.
The operator $\skappa$ nullifies any basis vector unless it has a non-zero (i.e., different from $e_0$) component in possitions specified by each pair $\{a_i,b_i\}$.  Since each basis vector in $\cH^{(m)}_k$ has exactly $k$ non-zero components, the first component of every basis vector used by $v$ must be $e_0$. Hence, $v\in \cH_0\otimes \bar \cH_k^{(m-1)}$.  By \rf{cor:kappa}, we also have $v\in \bar\cH_k^{(m)}$. Thus
\[
\Phi v = \bar\Pi_k^{(m)}v - \Pi_0\otimes \bar\Pi_k^{(m-1)}v - \Pi_1\otimes \bar\Pi_{k-1}^{(m-1)}v = v - v - 0 = 0,
\]
and $\Phi$ is zero on the isotypic subspace corresponding to $\lambda$.
This proves the second statement of \rf{lem:newMain}.

Now let us consider the remaining case: $\lambda_1 = m-k$.  
We use $\bbcH^{(m)}_k$ to denote the subspace of $\cH^{(m)}_k$ spanned by the irreps having exactly $k-1$ boxes below the first row.
In this notation, the subspace of $\cH^{(m)}_k$ spanned by the irreps $\specht\lambda$ of $\bS_{m-1}'$ with $\lambda_1 = m-k$ is given by $\cH_0\otimes\bbcH^{(m-1)}_k \oplus \cH_1\otimes\bar\cH^{(m-1)}_{k-1}$.  
We treat these two cases separately in the two claims below.
If $\Phi$ has norm $\OO(1/\sqrt{m})$ on both of them, then so it does on their direct sum, thus proving \rf{lem:newMain}.

\begin{clm}
\label{clm:odin}
$\Phi$ has norm $\OO(1/\sqrt{m})$ on $\cH_1\otimes\bar\cH^{(m-1)}_{k-1}$.
\end{clm}

\pfstart 
Fix $k-1$ disjoint pairs $\{a_1,b_1\},\ldots,\{a_{k-1},b_{k-1}\}\subset[2..m]$,
and let $a_k=1$ for notational convenience. 
Define an orthogonal projector
\begin{equation}
\label{eqn:kappaprim}
\skappa' = \frac{1}{2^{k-1}}\; \sA[\eps - (a_1,b_1)]\sA[\eps - (a_2,b_2)]\cdots \sA[\eps - (a_{k-1}, b_{k-1})].
\end{equation}
By Schur's Lemma and~\rf{lem:kappa}, there exists a principal right-singular vector $v\in\cH_1\otimes\bar\cH^{(m-1)}_{k-1}$ of 
$\Phi(\Pi_1\otimes\bar\Pi^{(m-1)}_{k-1})$ such that $\skappa'v=v$.

As $z=(z_1,\dots,z_k)$ runs through $[1..q-1]^k$, the following vectors
\[
w_z = 
\skD[
\frac{\w{a_1}{z_1}\w{b_1}{0}-\w{a_1}{0}\w{b_1}{z_1}}{\sqrt2} \otimes
\frac{\w{a_2}{z_2}\w{b_2}{0}-\w{a_2}{0}\w{b_2}{z_2}}{\sqrt2} \otimes\cdots\otimes
\frac{\w{a_{k-1}}{z_{k-1}}\w{b_{k-1}}{0}-\w{a_{k-1}}{0}\w{b_{k-1}}{z_{k-1}}}{\sqrt2}
]
\otimes
\w{a_k}{z_k}\w{}{0}\w{}{0}\cdots\w{}{0}
\]
form an orthonormal basis of the image of $\skappa'$ in $\cH_1\otimes\bar\cH^{(m-1)}_{k-1}$.
Hence, we can write $v=\sum_z\beta_z w_z$.

For brevity, let $u_z\in\cH^{\otimes(2k-2)}$ denote the unit vector in the square brackets in the above definition of $w_z$.
Then we can decompose $w_z = w'_z + w''_z$, with
\begin{align*}
&w'_z = \frac1{m-2k+2}\; u_z \otimes\sum_{b_k\in [m]\setminus\{a_1,b_1,\dots,a_{k-1},b_{k-1},a_k\}} \skB[\w{a_k}{z_k}\w{b_k}{0} - \w{a_k}{0}\w{b_k}{z_k}]
\otimes \w{}{0}\cdots\w{}{0},\\
&w''_z = \frac1{m-2k+2}\; u_z\otimes \sum_{c\,\in [m]\setminus\{a_1,b_1,\dots,a_{k-1},b_{k-1}\}} \w{c}{z_k}\w{}{0}\w{}{0}\cdots\w{}{0}.
\end{align*}
Note that $w'_z\in \bar\cH_k^{(m)}$ by \rf{lem:kappa}.
In addition, $v$ is orthogonal to $\cH_0\otimes\bar\cH^{(m-1)}_k$, therefore
\[
\Phi v = \bar\Pi^{(m)}_k v - v
= \sum_z\beta_z\sC[w'_z+\bar\Pi^{(m)}_k w''_z] - \sum_z\beta_z\sC[w'_z+w''_z] 
= -\sC[\I_\cH^{\otimes m}-\bar\Pi^{(m)}_k]\sum_z\beta_z w''_z.
\]
We have
 $\|w''_z\| = 1/\sqrt{m-2k+2}$ and the vectors $w''_z$ are mutually orthogonal.
 Thus
\[
\|\Phi v\|^2  
 \le \normC|\sum_z \beta_z w''_z|^2 = \sum_z \beta_z^2 \|w''_z\|^2 = \frac{1}{m-2k+2} \sum_z \beta_z^2 = \frac{1}{m-2k+2} \|v\|^2,
\]
and the norm of $\Phi$ on $\cH_1\otimes\bar\cH^{(m-1)}_{k-1}$ is at most $1/\sqrt{m-2k+2} = \OO(1/\sqrt m)$.
\pfend

\begin{clm}
\label{clm:dva}
$\Phi$ has norm $\OO(1/\sqrt{m})$ on $\cH_0\otimes\bbcH^{(m-1)}_k$.
\end{clm}

\pfstart
By Schur's Lemma there exist $\lambda\vdash m-1$ with $\lambda_1=m-k$ and a copy of the irrep $\specht\lambda$ in $\cH_0\otimes\bbcH^{(m-1)}_k$ consisting of principal right-singular vectors of $\Phi(\Pi_0\otimes\bbPi^{(m-1)}_k)$.
Let $t$ be a Young tableau of shape $\lambda$ with the entries in $[2..m]$.
Recall the operator $E_t = C_t^-R_t^+$ from~\rf{eqn:Et}.
There exists a non-zero vector $v$ in this irrep that is also in the image of $E_t$.
Let $d_1,\dots,d_\ell$ be the entries in the columns of $t$ of height $1$.  We have $\ell\ge m-2k+1$.
Also, for convenience, let $d_{\ell+1} = 1$.

Since $C_t$ does not affect $d_1,\dots,d_\ell$, 
and $R_t$ uses all the symmetric group on these elements,
the vector $v$ is symmetric with respect to permuting them.
Also, $v$ is in the image of $C^-_t$, which means it only uses basis vectors 
such that, for entries of each column of $t$, at most one corresponding component is $e_0$.
Consequently, these basis vectors have at most one non-zero component in positions $d_1,\dots,d_\ell$, and the vector $v$ is of the form
\[
v = u_0\otimes \w{d_1}{0}\w{d_2}{0}\cdots\w{d_{\ell}}{0}\otimes \w{d_{\ell+1}}{0}
+ \sum_{z\in[1..q-1]} u_z\otimes \sC[\sum_{j=1}^\ell \w{d_j}{z}\w{}{0}\cdots\w{}{0} ]\otimes \w{d_{\ell+1}}{0}
\]
for some vectors $u_0,u_1,\dots,u_{q-1}\in\cH^{\otimes (m-\ell-1)}$.
The vector $v$ can be decomposed as $v'+v''$ with
\begin{align*}
&
v' = u_0\otimes \w{d_1}{0}\w{d_2}{0}\cdots\w{d_{\ell+1}}{0}
+ \sum_{z\in[1..q-1]} u_z\otimes \sC[\sum_{j=1}^{\ell+1} \w{d_j}{z}\w{}{0}\cdots\w{}{0} ],
\\&
v'' = - \sum_{z\in[1..q-1]} u_z\otimes \w{d_{\ell+1}}{z}\w{}{0}\cdots\w{}{0}.
\end{align*}
Note the intentional reseblance between $v$ and $v'$, and note that
\[
\|v\|^2 \ge \sum_{z\in[1..q-1]} \ell \|u_z\|^2 = \ell \|v''\|^2.
\]
We claim that $v'$ is in the image of $E_{t'}$, where $t'$ is obtained from $t$ by appending the element 1 to the first row.
This means that $v'$ is contained in $\specht{\,\sh(t')}$.
In particular, $v'$ is orthogonal to $\bar\cH^{(m)}_k$.
Because $v$ is orthogonal to both $\cH_0\otimes \bar\cH_k^{(m-1)}$ and  $\cH_1\otimes \bar\cH_{k-1}^{(m-1)}$, we have
\[
\|\Phi v\|^2 
= \normA|\bar\Pi_k^{(m)}v|^2 = \normA|\bar\Pi_k^{(m)}v''|^2
\le \|v''\|^2 \le \|v\|^2/\ell.
\]
Hence, the norm of $\Phi$ on $\cH_0\otimes\bbcH^{(m-1)}_k$ is at most $1/\sqrt{m-2k+1} = \OO(1/\sqrt m)$.

Now let us prove our claim.  By linearity, it suffices to prove it when $v=E_te$ for some basis vector $e$ with $e_0$ as the first component.
Let $R_1$ be the set of entries in the first row of $t$, and $R'$ be the group of entry permutations within the rows of $t$ except the first one.  Denote by $A\subset R_1$ the set of entries in $R_1$ that correspond to non-zero components in $e$.
Then, $E_t e$ is proportional to
\[
\sum_{\substack{S\subseteq R_1:\; |S|=|A|\\ |S\cap\{d_1,\dots,d_{\ell}\}|\le 1}} \;\sum_{\substack{\pi\colon A\to S\\\text{$\pi$ is a bijection}}}
\; \sum_{\pi' \in R'} \sum_{\rho \in C_t}  (\sgn \rho) \rho\pi'\pi e,
\]
where $\pi e$ is defined in the obvious way,
and $E_{t'}e$ is proportional to
\[
\sum_{\substack{S\subseteq R_1\cup\{d_{\ell+1}\}:\; |S|=|A|\\ |S\cap\{d_1,\dots,d_{\ell+1}\}|\le 1}} \skD[\sum_{\substack{\pi\colon A\to S\\\text{$\pi$ is a bijection}}}
\; \sum_{\pi' \in R'} \sum_{\rho \in C_t}  (\sgn \rho) \rho\pi'\pi e].
\]
The only difference between the two expressions above is that there appeared new subsets $S$ in the outermost 
sum satisfying $S\cap\{d_1,\dots,d_{\ell+1}\} = \{d_{\ell+1}\}$.
Since neither $R'$ nor $C_t$ affect $\{d_1,\dots,d_{\ell+1}\}$,
the expression in the brackets with $S\cap\{d_1,\dots,d_{\ell+1}\} = \{d_{\ell+1}\}$
can be obtained from the same expression for $(S\setminus\{d_{\ell+1}\})\cup\{d_{\ell}\}$
by exchanging the $d_{\ell}$-th and the $d_{\ell+1}$-th copies of the space $\cH$ in the tensor product $\cH^{\otimes m}$.
This proves that $v'$, defined via the same vectors $u_0,u_1,\ldots,u_{q-1}$ as $v$, is proportional to $E_{t'}e$.
\pfend

\subsection{Proof of \rf{lem:normEstimations}}
\label{sec:normEstimationsProof}

Recall that the matrix $W_{\qp,k}$ is a block column matrix consisting of $|M_\qp|$ blocks $\frac1{\sqrt{|M_\qp|}} W_k^\mu$, one for every matching $\mu\in M_\qp$.  So, let us focus on $W_k^\mu$ given by~\rf{eqn:Fsigma}.
Although $W_k^\mu$ is defined on the whole $\cH^{\otimes 2n}$, in reality it maps $\cH^{(2n)}_k$ to $\cH^{(n)}_k$ and is zero on the orthogonal complement.

A convenient way to illustrate the action of $W_k^\mu$ is as follows.  Consider the graph $G_\mu$ on $2n$ vertices with $n$ edges given by the pairs in $\mu$.  The basis vectors of $\cH^{(2n)}_k$ correspond to the labellings of the vertices with components in $\{e_0,e_1,\dots,e_{q-1}\}$ that have exactly $k$ non-zero components.
Fix such a labelling and let $e$ be the corresponding basis vector.
If there is an edge connecting two non-zero components, then $W_k^\mu e=0$. 
Otherwise, $W_k^\mu e$ is a basis vector of $\cH^{(n)}_k$, which corresponds to the following labelling 
of the edges: an edge connecting two components $e_0$ is labeled by $e_0$ and an edge connecting $e_0$ and a non-zero component $e_i$ is labeled by $e_i$.

\begin{lem}
\label{lem:Wkv}
Recall the operator $\skappa$ from~\rf{eqn:kappa}.
For any vector $v\in\cH^{(2n)}_k$ satisfying $\skappa v = v$ and any matching $\mu$, we have $\|W_k^\mu v\|\le \|v\|$, $\|X_k^\mu v\|\le \|v\|$, and $\|Y_k^\mu v\|\le \|v\|$, 
where $W_k^\mu$, $X_k^\mu$, and $Y_k^\mu$ are defined in~\rfO(eqn:Fsigma) and~\rfO(eqn:GHI).
\end{lem}

\pfstart
We prove the result for $W_k^\mu$, the proofs for $X_k^\mu$ and $Y_k^\mu$ being similar.
Let, for brevity, $W = W_k^\mu$, and let $A_i = \{a_i,b_i\}$ be the pairs from the definition of $\skappa$.
Similarly as in the proof of \rf{clm:odin}, the following vectors give an orthonormal basis of the image of $\skappa$ as $z$ runs through $[1..q-1]^k$:
\[
w_z = 
\frac{\w{a_1}{z_1}\w{b_1}{0}-\w{a_1}{0}\w{b_1}{z_1}}{\sqrt2} \otimes
\frac{\w{a_2}{z_2}\w{b_2}{0}-\w{a_2}{0}\w{b_2}{z_2}}{\sqrt2} \otimes\cdots\otimes
\frac{\w{a_{k}}{z_{k}}\w{b_{k}}{0}-\w{a_{k}}{0}\w{b_{k}}{z_{k}}}{\sqrt2} \otimes
\w{}{0}\w{}{0}\cdots\w{}{0}.
\]
However, contrary to \rf{clm:odin}, the images of these vectors under $W$ are not orthogonal, thus we have to adopt a different proof strategy.

If $\mu$ contains a pair $A_i$ for some $i$, then $W v = 0$, so we will further assume this is not the case.
Construct the graph $G'$ from the graph $G_\mu$ by adding $k$ edges $a_1b_1,\dots,a_kb_k$.
Unlike the edges of $G_\mu$, these new edges will not be labeled, as they do not correspond to copies of $\cH$ in $\cH^{\otimes n}$.
This is a simple graph of maximal degree $2$, and it is a collection of even-length cycles and odd-length paths.  
The set of vertices of each cycle is a union of at least two pairs from $A_1,\dots,A_k$.
A path starts and ends outside of $\bigcup_i A_i$, but the set of its internal vertices is a union (possibly empty) of some of the pairs from $A_1,\dots, A_k$.
Let $\cC$ denote the set of connected components of $G'$, which from now on, to avoid confusion, we will call \emph{connected parts}.  For each $C\in\cC$, we identify $C$ with the set of its vertices.
We decompose the space $\cH^{\otimes 2n}$ as $\cH^{\otimes 2n}=\bigotimes_{C\in\cC} \cH^{\otimes C}$, where $\cH^{\otimes C}$ corresponds to the vectices in $C$.
The above discussion implies that each $w_z$ can be decomposed as $w_z = \bigotimes_{C\in\cC} w_{z,C}$ with $w_{z,C}\in \cH^{\otimes C}$ of norm $1$.

Our goal is to prove that $\|W \skappa\|\le 1$.
Let $\Xi_\skappa\colon \cH^{\otimes 2n}\to\cH^{\otimes 2n}$ be the orthogonal projector onto the span of the basis vectors used by some $w_z$.  These are the basis vectors having exactly one non-zero component in (possitions corresponding to) each $A_i$.
Let $\Xi_W\colon \cH^{\otimes 2n}\to\cH^{\otimes 2n}$ denote the orthogonal projector onto the span of the basis vectors \emph{not} nullified by $W$, i.e., having at most one non-zero component in each pair of $\mu$.
Since the image of $\skappa$ is contained in the coimage of $\Xi_\skappa$, and the coimage of $W$ is contained in the image of $\Xi_W$, we have $WF = W\Xi_W\Xi_\skappa\skappa$.

Note that both $\Xi_W$ and $\Xi_\skappa$ are diagonal $(0,1)$-matrices in the $e$-basis.
Let $\Xi = \Xi_W\Xi_\skappa$, which 
is the orthogonal projector onto the span of all the basis vectors that are used any $w_z$ and not mapped to 0 by $W$.
These basis vectors can be described in terms of $G'$.  In each cycle, vertices labeled by zero and non-zero components alternate.  In each path, each edge $a_ib_i$ has exactly one non-zero component (as its endpoint), and each edge of $G_\mu$ has at most one non-zero component. (See \rf{fig:graphExample} for an example of such graph.)
The conditions for distinct $C\in\cC$ are independent, so $\Xi$ can be decomposed as $\Xi = \bigotimes_{C\in\cC} \Xi_C$ with orthogonal projectors $\Xi_C\colon \cH^{\otimes C} \to\cH^{\otimes C}$.

\tikzstyle{fEdge} = [decorate,decoration={snake,amplitude=1.2pt,segment length=5pt,post length=0pt}]
\tikzstyle{mEdge} = []
\tikzstyle{cblue}=[circle, draw, thin,fill=cyan!20, scale=0.8]
\tikzstyle{sVer}=[circle, draw, thick, fill=white, scale=0.45] 
\tikzstyle{nVer}=[circle, draw,        fill=black, scale=0.4] 

\begin{figure}[!h]
\centering
\begin{tikzpicture}

  \draw[fEdge] (0,0)--(0,1);
  \draw[fEdge] (1,0)--(1,1);
  \draw[fEdge] (5,0)--(5,1);
  \draw[fEdge] (6,0)--(6,1);

  \draw[fEdge] (3,0)--(4,0);
  \draw[fEdge] (3,1)--(4,1);
  \draw[fEdge] (9,0)--(10,0);
  \draw[fEdge] (10,1)--(11,1);

  \draw (2,0)--(2,1);
  \draw (3,0)--(3,1);
  \draw (10,0)--(10,1);
  \draw (11,0)--(11,1);

  \draw (0,0)--(1,0);
  \draw (0,1)--(1,1);
  \draw (4,0)--(5,0);
  \draw (4,1)--(5,1);
  \draw (6,0)--(7,0);
  \draw (6,1)--(7,1);
  \draw (8,0)--(9,0);
  \draw (8,1)--(9,1);

  \foreach \place/\x in {
    {(0,0)/0}, {(1,0)/1}, {(2,0)/2}, {(3,0)/3}, {(4,0)/4}, {(5,0)/5},
    {(6,0)/6}, {(7,0)/7}, {(8,0)/8}, {(9,0)/9}, {(10,0)/10}, {(11,0)/11},
    {(0,1)/12}, {(1,1)/13}, {(2,1)/14}, {(3,1)/15}, {(4,1)/16}, {(5,1)/17},
    {(6,1)/18}, {(7,1)/19}, {(8,1)/20}, {(9,1)/21}, {(10,1)/22}, {(11,1)/23}}
  \node[nVer] (a\x) at \place {};

\node at (0, -0.35) {0};
\node at (1, -0.35) {5};
\node at (2, -0.35) {0};
\node at (3, -0.35) {0};
\node at (4, -0.35) {4};
\node at (5, -0.35) {0};
\node at (6, -0.35) {0};
\node at (7, -0.35) {0};
\node at (8, -0.35) {0};
\node at (9, -0.35) {1};
\node at (10, -0.35) {0};
\node at (11, -0.35) {0};

\node at (0, 1.35) {3};
\node at (1, 1.35) {0};
\node at (2, 1.35) {0};
\node at (3, 1.35) {7};
\node at (4, 1.35) {0};
\node at (5, 1.35) {4};
\node at (6, 1.35) {3};
\node at (7, 1.35) {0};
\node at (8, 1.35) {0};
\node at (9, 1.35) {0};
\node at (10, 1.35) {0};
\node at (11, 1.35) {2};

\node at (0.5,0.25) {\emph{5}};
\node at (4.5,0.25) {\emph{4}};
\node at (6.5,0.25) {\emph{0}};
\node at (8.5,0.25) {\emph{1}};

\node at (0.5,1.25) {\emph{3}};
\node at (4.5,1.25) {\emph{4}};
\node at (6.5,1.25) {\emph{3}};
\node at (8.5,1.25) {\emph{0}};

\node at (1.8,0.5) {\emph{0}};
\node at (2.8,0.5) {\emph{7}};
\node at (9.8,0.5) {\emph{0}};
\node at (10.8,0.5) {\emph{2}};

\end{tikzpicture}
\caption{An example of the graph $G'$ for $n=12$ and $k=8$, which has $K=2$ cycles.
The edges of $\mu$ are represented by straight lines and the edges $a_ib_i$ by wobbly curves.
For readibility, labels $e_z$ are shown as $z$ instead, and the labels of edges are in italics.} 
\label{fig:graphExample}
\end{figure}
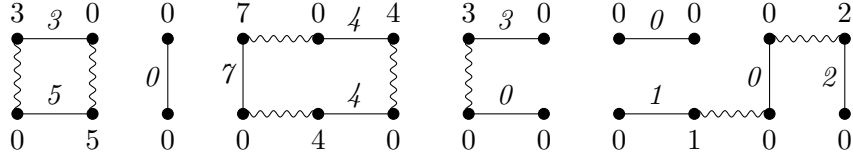

We have $\|W\skappa\| = \|W\Xi\skappa\| \le \|W\Xi\|\,\|\Xi \skappa\|$.
We bound both multipliers separately in the two claims below, from which the lemma follows.
Let $K$ denote the number of cycles in $G'$.

\begin{clm}
\label{clm:Lambdav}
We have $\| \Xi \skappa \| \le 1/\sqrt{2^K}$.
\end{clm}

\pfstart
First, $w_z$ form an orthonormal basis of the image of the projector $\skappa$.
Next, distinct $w_z$ use disjoint sets of basis vectors, hence $\Xi$ maps them to mutually orthogonal vectors.
Thus, it suffices to show that $\|\Xi w_z\|\le 1/\sqrt{2^K}$ 
for each $z$.
We have
\begin{equation}
\label{eqn:Lambdawz}
\|\Xi w_z\| = \normC|\bigotimes_{C\in\cC} \Xi_C w_{z,C} | = \prod_{C\in\cC} \norm|\Xi_C w_{z,C}|.
\end{equation}
If $C$ is a path, we conclude $\|\Xi_C w_{z,C}\| \le 1$ from $\Xi_C$ being a projector.
Now assume $C$ is a cycle.
As the length of $C$ is at least 4, the basis vectors of $\cH^{\otimes C}$ used by $w_{z,C}$ appear in $w_{z,C}$ with coefficients $1/\sqrt{2}^{|C|/2}\leq 1/2$ by the absolute value. But, since zero and non-zero components in a cycle must alternate, at most 2 of these basis vectors are in the image of $\Xi_C$.
Hence, $\|\Xi_C w_{z,C}\| \le1/\sqrt{2}$. 
Multiplying over all $C\in\cC$, we get the claim.
\pfend

\begin{clm}
\label{clm:WLambda}
We have $\| W \Xi \| \le {\sqrt{2^K}}$.
\end{clm}

\pfstart
Consider the matrix $M$ corresponding to $W\Xi$ written in the $e$-basis.
Its rows are labelled by the basis vectors of $\cH^{(n)}_k$, and its columns by the basis vectors used by $\Xi$.  
Because each basis vector used by $\Xi$ is mapped to a basis vector of $\cH^{(n)}_k$, each column of $M$ contains exactly one 1, and the remaining entries are all 0.  It suffices to prove that each row of $M$ contains at most $2^K$ non-zero entries.

Fix a labelling of the edges of $G_\mu$, and we can assume the corresponding row is non-zero.  
Consider a connected part $C\in\cC$.
If $C$ is a path, then there is a unique way to extend the labelling of the edges of $G_\mu$ to the vertices of $C$.  Indeed, one of the edges of $G_\mu$ in the path must have label $e_0$.  Both its endpoints are then $e_0$, and the whole labelling of the vertices of $C$ is uniquely determined.
If $C$ is a cycle, then there are exactly two possibilities to extend the labelling of the edges of $G_\mu$ to the vertices in $C$, because zero and non-zero components must alternate.
This proves that there are exactly $2^K$ entries equal to 1 in this row, and the remaining entries are 0.
\pfend

This ends the proof of \rf{lem:Wkv} for $W=W_k^{\mu}$.  The operators $X_k^{\mu}$ and $Y_k^\mu$ are similar to $W_k^\mu$ with the difference that they map more basis vectors to 0, and the bounds above hold for them as well.
\pfend

\begin{lem}
\label{lem:Wk-1}
Recall the operator $\skappa'$ from~\rf{eqn:kappaprim}.
For any vector $v\in\cH^{(2n)}_k$ satisfying $\skappa'v = v$ and any matching $\mu$, we have $\|W_k^\mu v\|\le \sqrt{3} \|v\|$, where $W_k^\mu$is defined in~\rfO(eqn:Fsigma).
\end{lem}

\pfstart
This is a modification of the proof of \rf{lem:Wkv}, and we adopt the notation from that proof.  We define $A_1,\dots,A_{k-1}$ as before.  The orthonormal basis of the image of $\skappa'$ consists of the vectors of the form
\[
w_{a_i,z} = \frac{\w{a_1}{z_1}\w{b_1}{0}-\w{a_1}{0}\w{b_1}{z_1}}{\sqrt2} \otimes\cdots\otimes
\frac{\w{a_i}{z_i}\w{b_i}{z_k}-\w{a_i}{z_k}\w{b_i}{z_i}}{\sqrt2} \otimes\cdots\otimes
\frac{\w{a_{k-1}}{z_{k-1}}\w{b_{k-1}}{0}-\w{a_{k-1}}{0}\w{b_{k-1}}{z_{k-1}}}{\sqrt2} \otimes
\w{}{0}\w{}{0}\cdots\w{}{0},
\]
where $z_i< z_k$, and
\[
w_{c,z} = \frac{\w{a_1}{z_1}\w{b_1}{0}-\w{a_1}{0}\w{b_1}{z_1}}{\sqrt2} \otimes\cdots\otimes
\frac{\w{a_{k-1}}{z_{k-1}}\w{b_{k-1}}{0}-\w{a_{k-1}}{0}\w{b_{k-1}}{z_{k-1}}}{\sqrt2} \otimes
\w{c}{z_k}\w{}{0}\cdots\w{}{0}.
\]
As before, if $\mu$ contains a pair $A_i$ for some $i$, then $W_k^\mu v=0$, so we assume this is not the case.  The graph $G'$, the set of connected parts $\cC$, and the projector $\Xi$ are all also defined as before.  There is still a decomposition $w = \bigotimes_{C\in\cC} w_C$ for each $w = w_{i,z}$. And the projector $\Xi$ still maps different $w_{i,z}$s to orthogonal vectors.

The main difference is in the description of the basis of the image of $\Xi$ in terms of $G'$.  In each cycle, zero and non-zero components still alternate.  However, there is now one \emph{special path}: the one containing $z_k$.
This path has a \emph{special place}: either a non-zero component outside of $\bigcup_i A_i$, or a pair $A_i$ with two non-zero components.  Starting from the special place, zero and non-zero components alternate.
(See \rf{fig:specialPath} for an example.)
All the other paths are as before.

\tikzstyle{fEdge} = [decorate,decoration={snake,amplitude=1.2pt,segment length=5pt,post length=0pt}]
\tikzstyle{mEdge} = []
\tikzstyle{cblue}=[circle, draw, thin,fill=cyan!20, scale=0.8]
\tikzstyle{sVer}=[circle, draw, thick, fill=white, scale=0.45] 
\tikzstyle{nVer}=[circle, draw,        fill=black, scale=0.4] 
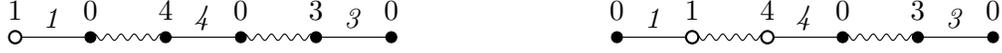
\begin{figure}[!h]
\centering
\begin{tikzpicture}

  \draw[fEdge] (1,3)--(2,3);  \draw[fEdge] (3,3)--(4,3);
  \draw[fEdge] (9,3)--(10,3);  \draw[fEdge] (11,3)--(12,3);

  \foreach \place/\x in {
                {(1,3)/13}, {(2,3)/23}, {(3,3)/33}, {(4,3)/43}, {(5,3)/53},
    {(8,3)/02},                         {(11,3)/32}, {(12,3)/42}, {(13,3)/52}}
  \node[nVer] (a\x) at \place {};

  \foreach \place/\i in {
    {(0,3)/03},
                {(9,3)/12}, {(10,3)/22}}
    \node[sVer] (a\i) at \place {};

  \path[thin] (a02) edge (a12); \path[thin] (a22) edge (a32); \path[thin] (a42) edge (a52);
  \path[thin] (a03) edge (a13); \path[thin] (a23) edge (a33); \path[thin] (a43) edge (a53);

\node at (0, 3.35) {1};
\node at (1, 3.35) {0};
\node at (2, 3.35) {4};
\node at (3, 3.35) {0};
\node at (4, 3.35) {3};
\node at (5, 3.35) {0};

\node at (8, 3.35) {0};
\node at (9, 3.35) {1};
\node at (10, 3.35) {4};
\node at (11, 3.35) {0};
\node at (12, 3.35) {3};
\node at (13, 3.35) {0};

\node at (0.5,3.25) {\emph{1}};
\node at (2.5,3.25) {\emph{4}};
\node at (4.5,3.25) {\emph{3}};

\node at (8.5,3.25) {\emph{1}};
\node at (10.5,3.25) {\emph{4}};
\node at (12.5,3.25) {\emph{3}};

\end{tikzpicture}
\caption{
Two types of special paths. Special places are denoted by one (left) or two (right) hollow circles.}
\label{fig:specialPath}
\end{figure}

Another difference is that the conditions on the basis vectors in $\Xi$ are no longer independent for different $C\in\cC$ (one path being special prevent other paths from being special). Hence,
no longer $\Xi$ decomposes into a tensor product.
However, there is a simple remedy.
Let $P\in\cC$ be a path, and $\Xi^P$ be the projector onto the span of the basis vectors of the image of $\Xi$ that have $P$ as their special path.
Now the conditions become independent, and we regain the decomposition $\Xi^P = \bigotimes_{C\in\cC} \Xi^P_C$.
Also note that all the basis vectors used by given $w_{i,z}$ share the same special path.
So, for a path $P\in\cC$, let $\skappa^P$ denote the projector onto the span of $w_{i,z}$ with special path $P$.

Let again $K$ be the number of cycles in $G'$.
We have the following analogues of Claims~\ref{clm:Lambdav} and~\ref{clm:WLambda}, respectively.
Here, $|P|$ denotes the number of vertices in path $P$, that is, its length plus $1$.

\begin{clm}
For a path $P\in\cC$, we have $\|\Xi^P \skappa^P\| \le 1/\sqrt{2^{K+\max\sfig{0, |P|-4}/2}}$.
\end{clm}

\pfstart
The projector $\Xi^P$ maps distinct $w_{i,z}$ to orthogonal vectors, hence it suffices to prove that $\|\Xi^P w\| \le 1/\sqrt{2^{K+\max\sfig{0, |P|-4}/2}}$ for any $w = w_{i,z}$ with special path $P$.

As in~\rf{eqn:Lambdawz}, we have $\|\Xi^P w\| = \prod_{C\in\cC} \|\Xi^P_C w_C\|$.
And we still use the bounds $\|\Xi^P_C w_C\|\le 1$ and $\|\Xi^P_C w_C\|\le 1/\sqrt2$ 
 when $C$ is a non-special path or a cycle, respectively.

Now consider $\Xi^P_P w_P$.  If $|P|=2$, then the bound $\|\Xi^P_P w_P\|\le \|w_P\|$ suffices.
So let us assume $|P|\ge 4$.  Then, $w_P$ uses $2^{(|P|-2)/2}$ basis vectors, and it is easy to check that at most $2$ of them are used in $\Xi^P$.
Indeed, the special position is uniquely determined, and, starting from it, zero and non-zero components must alternate.
The factor 2 comes from the two possible arrangements of non-zero components in the special pair $A_i$ (if there is one).
Thus, $\|\Xi^P_P w_P\|\le \sqrt{2^{(|P|-4)/2}} \|w_P\|$.
Multiplying over all $C\in\cC$, we get the required result.
\pfend

\begin{clm}
For a path $P\in\cC$, we have $\| W \Xi^P \| \le \sqrt{2^{K-1}(|P|+2)}$.
\end{clm}

\pfstart
Define the matrix $M$ for $W \Xi^P$ as in the proof of~\rf{clm:WLambda}.  
Still, each column of $W$ contains exactly one 1, the remaining entries all being 0.  
Fix a labelling of the edges of $G_\mu$ such that the corresponding row of $M$ is non-zero.

Consider a connected part $C\in\cC$.  Again, if $C$ is a non-special path or a cycle, then there are exactly 1 or 2 ways, respectively, to extend the labelling of the edges of $G_\mu$ to the vertices in $C$.
Now assume $C$ is the special path $P$.  All the edges of $G_\mu$ used by $P$ must have non-zero labels.  There are at most $(|P|+2)/2$ ways to extend this labelling to the labelling of the vertices in $P$: for each choice of the special position, there is at most one labelling.
Hence, there are at most $2^{K-1}(|P|+2)$ entries equal to 1 in each row of $M$, and the norm of $M$ is at most $\sqrt{2^{K-1}(|P|+2)}$.
\pfend

Now we can finish the proof of \rf{lem:Wk-1}.  
It suffices to prove that $\|WF'\|\le \sqrt3$.  
We have
$WF' = \sum_P WF^P = \sum_P W\Xi^PF^P$, where the sum is over all paths $P$ in $\cC$.
Note that $W$ maps basis vectors with distinct special paths to orthogonal vectors.
This means that $\|WF'\| = \max_P \|W\Xi^PF^P\|$.
Given any path $P\in\cC$, by the above two claims,
\[
\|W\Xi^PF^P\| \le \|W\Xi^P\|\, \|\Xi^P F^P\| \le 
\sqrt{\frac{(|P|+2)/2}{2^{\max\sfig{0, |P|-4}/2}}} \le \sqrt{3}.\qedhere
\]
\pfend
\medskip

Now we are ready to prove \rfO(lem:normEstimations).
Let us start with point (a) stating that $\|\bar X_{\qp,k}\|\le 1$, where $\bar X_{\qp,k} = X_{\qp, k} (\Pi_0\otimes \bar\Pi_{\qp,k}')$.
This matrix is symmetric with respect to $\bS'_\qp$.  Hence, Schur's lemma implies that there exists an irreducible $\bS'_\qp$-module all consisting of right-singular vectors of $\bar X_{\qp,k}$ of singular value $\|\bar X_{\qp,k}\|$.

By the definition of $\bar\Pi_{\qp,k}'$, the module is isomorphic to either $\specht{(2n-1-k,\lambda)}$ with $\lambda \vdash k$ for \cpp, or $\specht{(n-1-\ell,\lambda)}\otimes \specht{(n-k+\ell,\lambda')}$ with $\lambda\vdash \ell$ and $\lambda'\vdash k-\ell$ for \sep.  By \rfO(lem:kappa), in both cases, there exists a non-zero vector $v$ in the module satisfying $\skappa v = v$ for some choice of $a_1,b_1,\dots,a_k,b_k\in[2..2n]$ (in the case of \sep, one has to take the tensor product of two vectors obtained by two applications of \rfO(lem:kappa)).  
By \rfO(lem:Wkv), $\|X_k^\mu v\|\le \|v\|$ for all $\mu$. Hence $\|X_{\qp, k} v\|\le \|v\|$, and $\|\bar X_{\qp, k}\|\le 1$.

Consider point (b) now.  Similarly as for (a), we get a right-singular vector $v$ of singular value $\|\bar Y_{\qp, k}\|$ such that $\skappa v = v$.  Note that, if $\mu$ matches 1 with an element outside $\{a_1,b_1,\dots,a_k,b_k\}$, then, $Y_{\qp,k}^\mu v = 0$, because both these components are $e_0$ for all basis vectors used in $v$.
Otherwise, we still get $\|Y^\mu_{\qp, k} v\|\le 1$ by \rfO(lem:Wkv). The latter case only holds for an $\OO(k/n)$ fraction of all matchings, hence $\|\bar Y_{\qp ,k}\|= \OO(\sqrt{k/n})$ (cf. the third bullet in the proof of \rf{prp:learning}).

Now, let us prove (c).  We start with $W_{\cp,k}\Phi_{\cp,k}$.
Let $\Xi$ denote the projector onto the subspace of $\cH^{(2n)}_k$ spanned by the irreps of $\bS_\cp' = \bS_{[2..2n]}$ with exactly $k-1$ boxes below the first row.
From \rfO(lem:newMain), we know that $\normA|\Phi_{\cp,k}|=\OO(1/\sqrt{n})$ and $\Phi_{\cp,k} = \Xi \Phi_{\cp,k}$.  
So, it suffices to prove that $\|W_{\cp,k}\Xi \| = \OO(1)$.
This matrix is symmetric with respect to $\bS_\cp'$; hence, it has an irrep isomorphic to $\specht\lambda$ for some $\lambda \vdash 2n-1$ consisting of principal right-singular vectors.
Moreover, by the definition of $\Xi$, $\lambda$ has exactly $k-1$ boxes below the first row.
By \rf{lem:kappa}, $W_{\cp,k}\Xi$ has a principal right-singular vector that satisfies $\skappa'v=v$ with $\skappa'$ as in~\rf{eqn:kappaprim}.
\rf{lem:Wk-1} then implies that $\|W_{\cp,k} \Xi v\| = \|W_{\cp,k}v\|\le \sqrt3\|v\|$; hence, $\normA|W_{\cp,k}\Phi_{\cp,k}|=\OO(1/\sqrt{n})$.

Let us now consider the case of $W_{\se,k}\Phi_{\se, k}$.
Again, all these matrices are symmetric with respect to $\bS'_\se = \bS_{[2..n]}\times \bS_{[n+1..2n]}$.
Thus, there exist $\lambda \vdash n-1$, $\lambda'\vdash n$, and an irrep isomorphic to $\specht\lambda\otimes\specht{\lambda'}$ that consists solely of principal right-singular vectors of $W_{\se,k}\Phi_{\se, k}$.
By the definition of $\Phi_{\se,k}$ from~\rf{eqn:barHprime} and using \rf{lem:newMain}, we have that $\lambda$ has $k-\ell-1$ and $\lambda'$ has $\ell$ boxes below the first row for some $\ell\in[0..k-1]$.
By a double application of \rf{lem:kappa}, there exists a principal right-singular vector $v$ satisfying $F'v = v$.  We have
\[
W_{\se,k}\Phi_{\se, k} v = 
W_{\se,k}\sk[\sum\nolimits_{\ell=0}^{k-1}(\Phi^{(n)}_{k-\ell}\otimes \bar\Pi_{\ell}^{(n)})] v = 
W_{\se,k} (\Phi^{(n)}_{k-\ell} \otimes \bar\Pi^{(n)}_\ell) v.
\]
By \rf{lem:newMain}, $\|(\Phi^{(n)}_{k-\ell} \otimes \bar\Pi^{(n)}_\ell) v\| = \OO(1/\sqrt{n})\|v\|$.  Moreover, since $F'$, as a group algebra element, commutes with $\Phi^{(n)}_{k-\ell} \otimes \bar\Pi^{(n)}_\ell$, the vector $(\Phi^{(n)}_{k-\ell} \otimes \bar\Pi^{(n)}_\ell) v$ is also invariant under the action of $F'$.
Applying \rf{lem:Wk-1}, we get that $\|W_{\se,k}\Phi_{\se, k} v\| = \OO(1/\sqrt{n})\|v\|$, thus proving \rf{lem:normEstimations}.

\section{Discussion} \label{sec:concl}
In the paper, we prove tight lower bounds on the \cpp and \sep problems using the adversary method.
This is done by restricting the construction from~\cite{belovs:onThePower} on some specified irreducible representations of the symmetric group.

One oddity of this result is that it heavily uses the symmetry of the problem that is more typical for the polynomial method.  It is an interesting question whether representation theory of the symmetric group can be avoided in the proof.  In particular, consider the problem of distinguishing a 1-to-1 input from a 2-to-1 input whose matching, as defined in \rfO(sec:analysis), belongs to some specified set of matchings $M$.  Is it possible to characterise the quantum query complexity of this problem by a simple optimisation problem involving $M$?

Another open problem is whether the results in this paper can be combined with the results in~\cite{belovs:onThePower} in order to prove tight lower bounds for the {\sc $k$-distinctness} problem.  In particular, is the algorithm in~\cite{belovs:learningKDist} optimal?

\section*{Acknowledgments}
A.R. would like to thank Andris Ambainis, Chris Godsil, Robin Kothari, Laura Man{\v c}inska, and Robert {\v S}palek for fruitful discussions. 
A.B. was supported by the ERC Advanced Grant MQC and FP7 FET Proactive project QALGO. 
A part of this research was performed while A.R. was at the University of Waterloo supported by Mike and Ophelia Lazaridis Fellowship, David R. Cheriton Graduate Scholarship, and the US ARO.
A.R. also acknowledges the support of the Singapore National Research Foundation under NRF RF Award
No. NRF-NRFF2013-13.

\end{document}